\documentclass[final,journal,twoside]{IEEEtran}

\ifCLASSINFOpdf
  \usepackage[pdftex]{graphicx}
  \graphicspath{{figs/}}
  \DeclareGraphicsExtensions{.pdf,.jpeg,.png}
\else
  \usepackage[dvips]{graphicx}
  \graphicspath{{figs/}}
  \DeclareGraphicsExtensions{.eps}
\fi

\usepackage[cmex10]{amsmath}
\usepackage{amssymb,amsthm,amsfonts}
\usepackage{mathabx}
\usepackage{graphicx}
\usepackage{cite}
\usepackage{flushend}
\usepackage{hyperref}

\newtheorem{theorem}{Theorem}
\newtheorem{lemma}[theorem]{Lemma}

\theoremstyle{definition}
\newtheorem{definition}{Definition}
\theoremstyle{remark}
\newtheorem{remark}{Remark}
\newtheorem{example}{Example}

\newcommand{\nn}{\nonumber}
\newcommand{\integers}{\mathbb{Z}}

\newcommand{\const}[1]{{\mathcal{#1}}}
\newcommand{\Reals}{\mathbb{R}}
\newcommand{\Complex}{\mathbb{C}}

\newcommand{\E}{\mathsf{E}}

\newcommand{\bigo}[1]{\mathcal{O}\left(#1\right)}
\newcommand{\ie}{\emph{i.e.}}
\newcommand{\eg}{\emph{e.g.}}
\newcommand{\der}{\mathrm{d}}
\newcommand{\cc}{\mathrm{c.c.}}
\newcommand{\iid}{i.i.d.}
\newcommand{\eqdef}{\stackrel{\Delta}{=}}
\newcommand{\nr}{\textnormal{nr}}
\providecommand{\norm}[1]{\left\lVert#1\right\rVert}

\begin{document}

\title{The Kolmogorov-Zakharov Model\\for Optical Fiber Communication%
\thanks{
Published in the IEEE Transactions on Information Theory, vol. 63, no. 1, Jan. 2017, available
online at http://ieeexplore.ieee.org/document/7676314/. 
This work was supported in part by the Institute for Advanced
Study, Technische Universit\"at M\"unchen, funded by the German Excellence
Initiative and in part by the Alexander von Humboldt Foundation, funded by
the German Federal Ministry of Education and Research. 
\newline
\indent The author was with the Technische Universit\"at M\"unchen, 80333 Munich,
Germany. He is now with the Communications and Electronics
Department, T\'el\'ecom ParisTech, 75013 Paris, France (email: yousefi@telecom-paristech.fr).}}

\markboth{IEEE Transactions on Information Theory,  VOL. 63, NO. 1, JANUARY 2017}{M. I. Yousefi}

\author{Mansoor~I.~Yousefi}

\IEEEpubid{
\begin{minipage}{\textwidth}\ \\[12pt] \centering
0018--9448~\copyright~2016 IEEE. Personal use is permitted, but republication/redistribution requires IEEE permission. 
\\
See http://www.ieee.org/publications\_standards/publications/rights/index.html for more information.
\end{minipage}}

\date{}

\maketitle

\begin{abstract}
\ifCLASSOPTIONonecolumn\relax\else\boldmath\fi
A mathematical framework is presented to study the evolution of
multi-point cumulants in nonlinear dispersive partial differential
equations with random input data, based on the theory of weak wave
turbulence (WWT).   This framework is used to explain how energy is
distributed among Fourier modes in the nonlinear Schr\"odinger equation.
This is achieved by considering interactions among four Fourier modes
and studying the role of the resonant, non-resonant, and trivial
quartets in the dynamics.  As an application, a power spectral density
is suggested for calculating the interference power in dense
wavelength-division multiplexed optical systems, based on the
kinetic equation of the WWT. This power spectrum, termed the
Kolmogorov-Zakharov (KZ) model, results in a better estimate of the
signal spectrum in optical fiber, compared with the so-called Gaussian
noise (GN) model. The KZ model is generalized to non-stationary inputs
and multi-span optical systems. 
\end{abstract}

\begin{IEEEkeywords}
Fiber-optic communication,
weak wave turbulence, moments, cumulants, perturbation theory,
power spectral density.
\end{IEEEkeywords}

\maketitle


\section{Introduction}

\IEEEPARstart{T}{his} paper studies the power spectral density (PSD) and the
probability distribution of a signal propagating according to the
one-dimensional cubic nonlinear Schr\"odinger (NLS) equation, which serves as a
model for fiber-optic communication channels.  A PSD known as the Gaussian
noise (GN) model has been proposed for optical fiber communications
\cite{inoue1992phase,poggiolini2012gn}, resulting from a first-order
perturbation approach to four-wave mixing in the NLS equation
\cite{inoue1992phase}.  Although the GN model has appeared in the fiber-optic
communications literature, there also exists a satisfactory and well-developed
theory of the PSD for nonlinear dispersive equations in the context of wave
turbulence in mathematical physics, and this theory forms the foundation for
this paper.

From a mathematical point of view, the main idea of this approach to PSD
can be broadly abstracted as follows.  Suppose that a probability
measure at the input $z = 0$ of a partial differential equation (PDE)
describing signal evolution in distance $z$ is given. We are interested
in finding the probability measure at some $z>0$.  First, all $n$-point
moments of this probability distribution at $z = 0$ are found. For
reasons that are explained in Section~\ref{sec:energy-mechanisms}, it is
convenient to work with cumulants in place of moments, which are in
one-to-one relation with one another.  Then, a hierarchy of differential
equations is obtained that governs the evolution of cumulants in
distance. Under certain assumptions, this hierarchy is truncated at some
order $n$, ignoring the influence of higher order cumulants. This turns
the original infinite-dimensional functional problem into a
finite-dimensional differential system for a set of scalar parameters,
which is subsequently solved to obtain cumulants at distance $z > 0$.
Finally,  cumulants at $z$ are combined to obtain the probability
measure at $z$. The aim of wave turbulence theory is to study energy
distribution in the frequency domain via the PSD, a 2-point cumulant. As
a result, the emphasis is heavily placed on correlation function;
nevertheless the idea is useful to obtain information on higher-order
statistics.

In strong turbulence, encountered \eg, in the Navier-Stokes equations of
hydrodynamics,  nonlinearity can be strong and the cumulant hierarchy may
not truncate.  However in weak wave turbulence (WWT), under a weak
nonlinearity assumption, which holds in optical fiber, the hierarchy is
truncated and a differential equation for the PSD, known as a
\emph{kinetic equation}, is obtained.
\cite{zakharov1992turb,zakharov2009tis,zakharov2004owt}.  The kinetic
equations of WWT can often be solved using, \eg, Zakharov conformal
transformations to obtain Kolmogorov-Zakharov (KZ) stationary spectra.
In this paper, we study 2-, 4- and 6-point cumulants, with the aim of
obtaining a more accurate spectrum for interference and probability
distribution than what is currently known in communications. 

\IEEEpubidadjcol

Turbulence theory helps us to understand the mechanisms by which the
signal of one user is transported to the other users in a multiuser
communication system. For instance, kinetic equations provide accurate
predictions of interference power. More importantly, the theory gives
useful insights into inter-channel interactions in wavelength-division
multiplexing (WDM). WWT also yields a suitable mathematical framework
for statistical signal analysis in optical fiber, which we use for
modeling.

The contributions of the paper are organized as follows.

In
Section~\ref{sec:channel-model} we describe the channel model. We focus
on the frequency domain and introduce both a continuous and a discrete
model, as it turns out there are differences among them.

In Section~\ref{sec:gn-model} we  derive the GN model in a simplified
manner. We point out some of the shortcomings of the perturbation
approach when applied to equations of type NLS.

To help understand interference, in Section~\ref{sec:energy-transfer} we describe energy transfer mechanisms in optical
fiber via resonant manifolds and classification of quartets. We explain
how energy transport differs between integrable and non-integrable channels.

We introduce the basic KZ model in Section~\ref{sec:kz-model} and explain
how it relates to the GN model. The standard kinetic equation of the WWT
predicts a stationary spectrum for the integrable NLS equation. However, the
analysis
can be carried out to the next order in the nonlinearity
level to account for deviations from the stationary spectrum. 
We compare the KZ and GN PSDs
and show that the KZ PSD is equally simple yet provides better
estimates of WDM interference. The KZ model describes energy fluxes
correctly; for instance, unlike the GN model, the KZ spectrum of the lossless optical fiber is
energy-preserving. The KZ model also
predicts a quasi-Gaussian distribution, which is close to a Gaussian
one. However, this small deviation from the Gaussian distribution, not
captured by the GN model, is responsible
for spectrum evolution and interference. 

Sections~\ref{sec:wdm-application} and \ref{sec:multi-span} are dedicated to further details
about the KZ and GN models. The assumptions of the KZ model are
examined in the context of communications. The KZ model
is generalized to WDM and multi-span systems. 

The power spectral density is a central object in statistical studies of
nonlinear dispersive waves and is widely studied in mathematical
physics\footnote{In physics, the PSD can often be recognized in relation
with terms wave number, wave-action density, particle number, occupation
number, pair correlator, correlation function, energy density, etc.}.
One aim of this paper is to point this out and show that some of the
elaborate PSD calculations in communications engineering can be
succinctly modified and methodically generalized in the more fundamental
framework of WWT. The reader is referred to
\cite{zakharov1992turb,zakharov2004owt} for an introduction to WWT and
to \cite{turitsyn2013owt, picozzi2014owt} for a survey of optical
turbulence---mostly in higher-dimensional non-integrable models or in
laser applications \cite{turitsyn2013owt}. One-dimensional turbulence in
integrable systems, which is the focus of this paper, is discussed in
\cite{zakharov2009tis, suret2011wave}.


\section{Notation and Preliminaries}
\label{sec:notation}

The Fourier transform of $q(t)$ is represented as 
\begin{IEEEeqnarray}{rCl}
  \mathcal{F}(q)(\omega)\eqdef\int\limits_{-\infty}^{\infty}q(t)e^{j\omega t}\der t.
\label{eq:ft}
\end{IEEEeqnarray}
We use subscripts to denote the frequency variable, \eg,
$q_{\omega}=\mathcal{F}(q)(\omega)$. Fourier series coefficients of a
periodic signal $q(t)$ are similarly denoted by
$q_k=\mathcal{F}_s(q)(k)$. When there are multiple frequencies
$\omega_i$ in an expression, for brevity we often use shorthand
notations $q_i\eqdef q(\omega_i)$ and $\der\omega_{1\cdots
  n}\eqdef\prod_{i=1}^n\der\omega_i$. In such cases, it will be clear from
the context whether $q_i$ corresponds to
a discrete or a continuous frequency variable. To avoid confusion, we do not use
subscripts to denote a time variable. 

The following notation is used throughout the paper 
\begin{IEEEeqnarray*}{rCl}
\delta_{k_1\cdots k_{2n}}\eqdef\delta(s_1k_1+\cdots+s_{2n}k_{2n}),\quad k_i\in\integers,
\end{IEEEeqnarray*}
where $\delta(m)$ is the Kronecker delta and 
\begin{IEEEeqnarray*}{rCl}
s_i\eqdef 
\begin{cases}
1, & 1\leq i\leq n,\\
-1, & n+1\leq i\leq 2n.
\end{cases}
\end{IEEEeqnarray*}
For continuous
frequencies, the corresponding real subscript $\omega_1\cdots \omega_{2n}$ is shortened to
the integer 
subscript $1\cdots 2n$ 
\begin{IEEEeqnarray}{rCl}
\delta_{1\cdots2n}\eqdef\delta(s_1\omega_1+\cdots+s_{2n}\omega_{2n}),\quad\omega_i\in\Reals,  
\IEEEeqnarraynumspace
\end{IEEEeqnarray}
where, with notation abuse, $\delta(\omega)$ is the Dirac delta
function.

Let $q(t)$ be a zero-mean stochastic process. 
The symmetric $2n$-point \emph{correlation functions} in time (temporal moments) are
\begin{IEEEeqnarray}{rCl}
  R(t_1\cdots t_{2n})\eqdef
\E \Bigl[q(t_{1})\cdots q(t_{n})q^*(t_{n+1})\cdots
q^*(t_{2n})\Bigr],
\label{eq:n-point-correlation}
\end{IEEEeqnarray}
where $\E$ denotes expectation with respect to the
corresponding joint probability distribution. Similarly, the
$2n$-point \emph{spectral moments} are
\begin{IEEEeqnarray}{rCl}
\mu_{1\cdots 2n}\eqdef\E\Bigl[ q_1\cdots q_{n}q^*_{n+1}\cdots
q^*_{2n}\Bigr].
\label{eq:n-point-moment}
\end{IEEEeqnarray}
The asymmetric correlation functions and moments, in which the number of
conjugate and non-conjugate variables is not the same, is assumed to
be zero. If $n=1$, $\mu_{12}$ corresponds to the correlation between
$q(\omega_1)$ and $q(\omega_2)$.

If a stochastic process $q(t)$ is (strongly) stationary, then 
\begin{IEEEeqnarray}{rCl}
  R(t_1,\cdots, t_{2n})=R(t_1-t_0,\cdots, t_{2n}-t_0),
\label{eq:wss}  
\end{IEEEeqnarray}
for any reference point $t_0$ and $n\geq 1$. It is shown in 
Appendix~\ref{app:cumulants-wss} that if $q(t)$ is stationary, then 
\begin{IEEEeqnarray}{rCl}
\mu_{1\cdots 2n} = S_{1\cdots 2n}\delta_{1\cdots 2n},
\label{eq:n-point-psd}
\end{IEEEeqnarray}
where
\begin{IEEEeqnarray*}{rCl}
  S_{1\cdots
    2n}&=&\mathcal{F}(R(0,t_2,\cdots,t_{2n}))(0,s_2\omega_2, \cdots,s_{2n}\omega_{2n}),
\end{IEEEeqnarray*}
is the \emph{moment density function}.
Thus $\mu_{1\cdots 2n}$ is non-zero only on the \emph{stationary manifold}
\begin{IEEEeqnarray}{rCl}
  s_1\omega_1+\cdots+s_{2n}\omega_{2n}=0.
\label{eq:stationary-manifold}
\end{IEEEeqnarray}
If $n=1$, $\mu_{12}=S_{11}\delta_{12}$. We shorten equal indices as $S_{k}\eqdef S_{kk}$.

In addition to correlation functions and spectral moments, we also
require spectral cumulants $\kappa_{1\cdots 2n}$ and their densities $\tilde S_{1\cdots 2n}$. The reader is referred to
Appendix~\ref{app:cumulants-moments} for the definition of cumulants and their
relation with moments. It is shown that the $2n$-point
moments decompose in terms of the $2k$-point cumulants, $k\leq n$. Particularly, if $n=2, 3$, 
from \eqref{eq:cumulant-to-moment-6}:
    \begin{IEEEeqnarray}{rCl}
      \mu_{1234} &=& S_{1234}\delta_{1234}\nn\\
&=&
S_1S_2
\Bigl(\delta_{13}\delta_{24}+\delta_{14}\delta_{23}\Bigr)+\tilde{S}_{1234}\delta_{1234},
\label{eq:4-point}
  \\
\mu_{123456}&=&S_{123456}\delta_{123456}
\nn 
\\
&=&S_1S_2S_3\Bigl(\delta_{14}\delta_{25}\delta_{36}
+
\delta_{14}\delta_{26}\delta_{35}
\nn
\\
&&
+
\delta_{15}\delta_{24}\delta_{36}
+
\delta_{15}\delta_{26}\delta_{34}
\nn
\\
&&
+
\delta_{16}\delta_{24}\delta_{35}
+
\delta_{16}\delta_{25}\delta_{34}
\Bigr)+\tilde{S}_{123456}\delta_{123456},
\IEEEeqnarraynumspace
\label{eq:6-point}
\end{IEEEeqnarray}
where $\tilde S_{1234}$ and $\tilde S_{123456}$ are cumulant densities.

For a Gaussian distribution, only the mean and the 2-point
cumulants are non-zero. Consequently, $\mu_{1\cdots 2n}$ is concentrated on \emph{normal manifolds}
\begin{IEEEeqnarray*}{rCl}
  s_l\omega_{l}+s_k\omega_{k}=0, \quad 1\leq l\leq n,\quad n+1\leq k\leq 2n,
\end{IEEEeqnarray*}
which are subsets of the stationary manifold. Other distributions
generally
have infinitely many
non-zero cumulants. As a result, cumulants are used in this paper to measure
deviations from the Gaussian distribution. A zero-mean distribution is defined to be
\emph{quasi-Gaussian} if 
\begin{IEEEeqnarray}{rCl}
  \tilde S_{1\cdots n} \approx 0, \quad \forall n\geq 6.
\label{eq:quasi-gaussian}
\end{IEEEeqnarray}  
That is to say, at most 4-point cumulants are significant. 

We will often make use of the trilinear integral and sum of the signals
$q_{\omega}(z)$ and $q_k(z)$, defined as
\begin{IEEEeqnarray}{rCl}
  \mathcal{N}_{\omega}(q,q,q)(z)\eqdef
  \int\limits_{-\infty}^\infty
  q_{1}(z)q_{2}(z)q^*_{3}(z)
\delta_{123\omega}\der \omega_{123},
\IEEEeqnarraynumspace
\label{eq:N-omega}
\end{IEEEeqnarray}
and
\begin{IEEEeqnarray}{rCl}
\mathcal{N}_k(q,q,q)(z)\eqdef \sum\limits_{lmn\:\in\: \nr_k}
q_\ell(z)q_m(z)q_n^*(z),
\label{eq:N-k}
\end{IEEEeqnarray}
where $\nr_k$ is the set of the non-resonant frequencies
\begin{IEEEeqnarray*}{rCl}
  {\nr}_k\eqdef\Bigl\{ (l,m,n)\:\bigl|\: l+m=n+k,\: l\neq k, \: m\neq k\Bigr\}.
\end{IEEEeqnarray*}
These expressions help to factor out part of the
complexity. 

The following simple lemma is frequently used in Section~\ref{sec:gn-model} when passing from the zero-order to the
first-order in perturbation expansions.
\begin{lemma}
Let $q_\omega(z)\eqdef\exp\left(j\omega^2 z\right)q_\omega(0)$. Then
\begin{IEEEeqnarray*}{rCl}
  \int\limits_{0}^z
  e^{-jz'\omega^2}\mathcal{N}_{\omega}(q,q,q)(z')\der
  z'&=&\nn
\\ && \hspace{-3cm}j\int H_{123\omega}(z) q_1(0)q_2(0)q_3^*(0)\delta_{123\omega}\der\omega_{123\omega},  
\end{IEEEeqnarray*}
where the $H$-function is
\begin{IEEEeqnarray}{rCl}
H_{123\omega}(z)\eqdef
\begin{cases}
(1-e^{j\Omega_{123\omega}
  z})/\Omega_{123\omega}, & \Omega_{123\omega}\neq 0,\\
-jz, & \Omega_{123\omega}=0,
\end{cases}
\IEEEeqnarraynumspace   
\label{eq:H}
\end{IEEEeqnarray}
in which
\begin{IEEEeqnarray*}{rCl}
  \Omega_{123\omega}\eqdef\omega_1^2+\omega_2^2-\omega_3^2-\omega^2.
\end{IEEEeqnarray*}
\label{lemm:zero-to-one}
\end{lemma}

\begin{proof}
The result  follows by substitution.
\end{proof}

Depending on the context, we may write the $H$-function as  $H_{123\omega}$,
$H_{123\omega}(z)$ or $H(\Omega_{123\omega})(z)$. A similar lemma can
be stated for the trilinear sum \eqref{eq:N-k}. 


\section{Channel Model}
\label{sec:channel-model}

\subsection{Continuous-frequency NLS Equation}

We consider the one-dimensional cubic dimensionless NLS equation on the real line
\begin{IEEEeqnarray}{rCl}
  j \partial_z q=\partial_{tt}q+2|q|^2q,\quad (t,z)\in \Reals\times\Reals^+,
\label{eq:cnls}
\end{IEEEeqnarray}
where $q(t,z)$ is the signal as a function of space $z$ and time
$t$. To focus on main ideas, in this section we consider only a
single span lossless fiber.  Loss and amplification in
multi-span systems are introduced later in Section~\ref{sec:multi-span}. 

Using Duhamel's formula, the differential equation \eqref{eq:cnls} can
be re-written as an integral equation 
\begin{IEEEeqnarray}{rCl}
  q_{\omega}(z)&=&e^{jz\omega^2}q_{\omega}(0)\nn\\
&&-\:2j\int\limits_{0}^z
e^{j(z-z')\omega^2}\mathcal{N}_{\omega}(q,q,q)(z')\der z',
\label{eq:cnls-integ}
\end{IEEEeqnarray}
where $\mathcal{N}_\omega$, defined in \eqref{eq:N-omega}, represents
interaction among all four waves $123\omega$, and $\delta_{123\omega}$
denotes the corresponding frequency
matching condition 
\begin{IEEEeqnarray}{rCl}
 \omega_1+\omega_2=\omega_3+\omega.
\label{eq:hyperplane}
\end{IEEEeqnarray}
An interacting quartet can be shown schematically as $12\rightarrow 3\omega$. 

\begin{definition}[Trivial interactions]
\label{rem:trivial-ints}
A subset of frequencies in \eqref{eq:hyperplane} are trivial interactions
\begin{IEEEeqnarray}{rCl}
  \left(\omega_1=\omega_3, \:\omega_2=\omega\right),\quad 
 \left(\omega_1=\omega,\: \omega_2=\omega_3\right).
\label{eq:trivial-int}
\end{IEEEeqnarray}
These frequencies form a set of zero Lebesgue measure on the hyperplane \eqref{eq:hyperplane} and do not
contribute to the integral \eqref{eq:cnls-integ} --- unless the integrand has a delta function 
on \eqref{eq:trivial-int}; see \eqref{eq:mu-123'1'2'3}.
\qed
\end{definition}

\subsection{Discrete-frequency NLS Equation}

We also consider the NLS equation on torus $t\in \mathbb{T}=\Reals/(T\mathbb{Z})$, corresponding to $T$-periodic
signals. Partitioning the sums
\begin{IEEEeqnarray}{rCl}
  \sum\limits_{lm}&=&\sum\limits_{(l=k)\vee (m=k)}+\sum\limits_{(l\neq
    k) \wedge (m\neq k)},\nn\\
  \sum\limits_{(l=k)\vee (m=k)}&=&\sum\limits_{l=k}+\sum\limits_{m=k}-\sum\limits_{(l=k)\wedge (m=k)},
\label{eq:sum-or-and}
\end{IEEEeqnarray}
where $\vee$ and $\wedge$ are, respectively, \emph{or} and \emph{and} operations,
we get the identity
\begin{IEEEeqnarray}{rCl}
  \mathcal{F}_s(|q|^2q)(k)=2\const{P}q_k-|q_k|^2q_k+
\mathcal{N}_k(q,q,q)(z),
\label{eq:identity}
\end{IEEEeqnarray}
where $\const{P}\eqdef\norm{q(t)}_2^2/T$. The NLS equation in the discrete frequency domain is
\begin{IEEEeqnarray}{rCl}
   \partial_z
   q_k=j\omega_0^2k^2q_k-\underbrace{4j\const{P}q_k}_{\text{XPM}}+\underbrace{2j|q_k|^2q_k}_{\text{SPM}}
-2j\underbrace{\mathcal{N}_k(q,q,q)(z)}_{\text{FWM}},
\IEEEeqnarraynumspace
\label{eq:dnls}
\end{IEEEeqnarray}
where SPM, XPM and FWM denote self-phase modulation, cross-phase
modulation and four-wave mixing. 
Note that the SPM and XPM 
indices ($l=k$  or $m=k$) have been removed from $\mathcal N_k$. Unlike their
continuous version \eqref{eq:trivial-int},  these indices form a set
with non-zero measure and have no analogue in
\eqref{eq:cnls-integ}. Note further that the XPM is a constant phase
shift, thanks to conservation of energy. 

\begin{figure}
\centering
\begin{tabular}{c@{\hskip 0.5cm}c}
\includegraphics{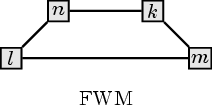} &
\includegraphics{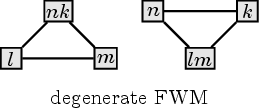} \\
(a)& (b) \\[0.3cm]
\includegraphics{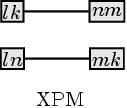} &
\includegraphics{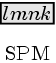} \\
(c) & (d)
\end{tabular}
\caption{Interacting quartets. Here $(l,m)$ and conjugate variables $(n,k)$ are
  shown on two copies of the $x$-axis, placed on top of each
  other. A combined index like $nk$ means $n=k$. (a) Regular FWM
  $x_lx_mx_n^*$ ($l\neq m\neq n\neq k$), (b)
  two degenerate FWMs $x_lx_mx_k^*$ ($l\neq m\neq k$) and $x_l^2x_n^*$
  ($l\neq n\neq k$),
  (c) two XPMs $|x_l|^2x_k$ and $|x_m|^2x_k$ ($l, m\neq k$) (d) SPM
$|x_k|^2x_k$ ($l=m=n=k$).}
\label{fig:quartets}
\end{figure}

As in the continuous model, the integral form of \eqref{eq:dnls} is
\begin{IEEEeqnarray}{rCl}
 q_k(z)&=&e^{j(\omega_0^2k^2-4\const{P}) z}\Bigl\{q_k(0)-2j\int_0^z
  e^{-j(\omega_0^2k^2-4\const{P})z'}\nn \\
&&\:\times\Bigl(-|q_k(z')|^2q_k(z')+
 \mathcal{N}_k(q,q,q)(z')\Bigr)\der z'\Bigr\}.
\IEEEeqnarraynumspace
\label{eq:dnls-integ}
\end{IEEEeqnarray}

\begin{example}[Classification of quartets]
Consider the sum
\begin{IEEEeqnarray*}{rCl}
S&\eqdef& \bigl|\left(x_{-2}+x_{-1}+x_{0}+x_{1}+x_{2}\right)\bigr|^2\nn\\
&&\times \left(x_{-2}+x_{-1}+x_{0}+x_{1}+x_{2}\right).
\end{IEEEeqnarray*}
The interference terms at frequency $k=0$ are ($l+m=n+0$)
\begin{IEEEeqnarray}{rCl}
S_0&=&
|x_0|^2x_0+\underbrace{2x_0\left[|x_{-2}|^2+|x_{-1}|^2+|x_{1}|^2+|x_{2}|^2\right]}_{\text{degenerate quartet } (l=k) \vee (n=k)}\nn\\
&&+\Bigl\{\underbrace{2x_0^*\bigl(
x_{-2}x_2+x_{-1}x_1\bigr)}_{\text{degenerate FWM }n=k}
+\underbrace{\bigl(x_{-1}^2x_{-2}^*+x_{1}^2x_{2}^*\bigr)}_{\text{degenerate FWM }l=m}\nn\\
&&+\underbrace{2\bigl(
x_{-2}x_{-1}^*x_{1}+
+x_{-1}x_{1}^*x_{2} 
\bigr)}_{\text{non-degenerate FWM }l\neq m\neq n\neq k}\Bigr\}.
\label{eq:example}
\end{IEEEeqnarray}
There are several possibilities for a quartet $lm\rightarrow nk$. If all indices are different, 
we get non-degenerate FWM. If two indices
of the same conjugacy type are equal, \ie, if $l=m$ or $n=k$, we obtain degenerate FWM. 
The cases that two indices of the opposite conjugacy type
are the same, \ie, $l=k$ or $n=k$, are also degenerate quartets. These are the terms with the square brackets in \eqref{eq:example}. 
The literature refers to these terms as 
XPM, not degenerate FWM. The degenerate quartet with multiplicity two where $l=m=n=k$ is known as the SPM
in the literature. However, according to our definition in \eqref{eq:dnls}, the SPM in \eqref{eq:example} is $-|x_0|^2x_0$ and the 
XPM is the term with the square brackets plus $2|x_0|^2x_0$. This simplifies XPM to
$2\const{P}x_0$ and negates the sign of the SPM, as in
\eqref{eq:identity}--\eqref{eq:dnls}.     
There are one SPM, ten XPM and ten FWM terms in this example. In general
if $-N \leq k \leq N$, 
a simple counting shows that there are $3N^2+3N+2$ (XPM and FWM) interference terms
at $k=0$. This number decreases as $k$ approaches the boundaries
$\pm N$. 
The XPM, degenerate and non-degenerate FWM
constitute, respectively,  the 1-, 2- and 3-wave interference.  

\qed
\label{ex:interference}
\end{example}


\section{GN Model}
\label{sec:gn-model}
The GN ``model'' in the literature refers to a PSD. In this
section, we re-derive this PSD for the continuous and
discrete models in a simplified manner. This clarifies GN PSD, so
that in Section \ref{sec:kz-vs-gn} it can be compared with the KZ PSD.

\subsection{Continuous-frequency NLS Equation}

Note that \eqref{eq:cnls-integ} is a fixed-point equation, mapping
$q(z)$ to itself. Iterating the fixed-point map
$q^{(k)}_\omega(z)\rightarrow q^{(k+1)}_\omega(z)$ starting from $q^{(-1)}_\omega=0$, we obtain
\begin{IEEEeqnarray*}{rCl}
  q_{\omega}^{(0)}(z)=e^{jz\omega^2}q_{\omega}(0).
\end{IEEEeqnarray*}
This is just the solution of the linear part of the NLS equation.
Iterating one more time and using Lemma~\ref{lemm:zero-to-one}, the
signal to the first-order in nonlinearity level is
\begin{IEEEeqnarray}{rCl}
  q_{\omega}^{(1)}(z)&=&
e^{jz\omega^2}\Bigl\{q_{\omega}(0) \nn\\
&&\hspace{-1cm}+\:2 \int H(\Omega_{123\omega})(z)q_1(0)q_2(0)q_3^*(0)\delta_{123\omega}\der\omega_{123}\Bigr\},
\IEEEeqnarraynumspace
\label{eq:q-pert}
\end{IEEEeqnarray}
where $H(\Omega_{123\omega})(z)$ is defined in \eqref{eq:H}.

It follows that the NLS equation has the simple closed-form solution
\eqref{eq:q-pert} to the 
first-order in the perturbation expansion. As a consequence, derived quantities such as
the PSD can also be calculated. Computing $\mu_{12}$ from \eqref{eq:q-pert} and removing factor $\delta_{12}$, we get
\begin{IEEEeqnarray}{rCl}
S_{\omega}(z)&=&S_{\omega}(0)+4\int \Re(H_{123\omega}S_{123\omega}) \delta_{123\omega}\der\omega_{123} \nn\\
&&\hspace{-2em}+4\int H_{123\omega}H_{1'2'3'\omega}^* S_{123'1'2'3}
\delta_{123\omega}\delta_{1'2'3'\omega}\der\omega_{1231'2'3'}.
\IEEEeqnarraynumspace
\label{eq:s-pert}
\end{IEEEeqnarray}

Equation \eqref{eq:s-pert} expresses a 2-point PSD as a function
of the 4- and 6-point PSDs. We can close the
equation for the 2-point PSD if we assume that signal statistics are
Gaussian. With this assumption, the 4- and 6-point PSDs break down according to \eqref{eq:4-point}--\eqref{eq:6-point}, with
zero cumulants. From \eqref{eq:4-point}
\begin{IEEEeqnarray}{rCl}
S_{123\omega}\delta_{123\omega}=S_1S_2(\delta_{13}\delta_{2\omega}+\delta_{1\omega}\delta_{23}).
\label{eq:S-is-real}
\end{IEEEeqnarray}
The right hand side in \eqref{eq:S-is-real} is real and supported on trivial interactions 
\eqref{eq:trivial-int}, where $H_{123\omega}=-jz$. Thus 
$\Re(H_{123\omega}S_{123\omega})=0$ and the first integral in \eqref{eq:s-pert} vanishes. 

For the second integral in \eqref{eq:s-pert}, note that
\begin{IEEEeqnarray*}{rCl}
 S_{123'1'2'3}\delta_{1'2'3'\omega}\delta_{123\omega}&=&
 S_{123'1'2'3}\delta_{123'1'2'3}\delta_{123\omega}
\\
&=&
 \mu_{123'1'2'3}\delta_{123\omega}.
\end{IEEEeqnarray*}
From \eqref{eq:6-point}
\begin{IEEEeqnarray}{rCl}
 \mu_{123'1'2'3}=
S_1S_2S_3\left(\delta_{11'}\delta_{22'}+\delta_{12'}\delta_{21'}\right)\delta_{33'},
\label{eq:mu-123'1'2'3}
\end{IEEEeqnarray}
where the other four terms are ignored. They lead to secular terms; we will
include them in Section~\ref{sec:secular}, \eqref{eq:S-sec}. 

Integrating over primed variables, the resulting
first-order PSD is
\begin{IEEEeqnarray}{rCl}
   S_{\omega}^{\text{GN}}(z)=S_\omega^0+8
\int
 |H_{123\omega}|^2S_{1}^0S_{2}^0S_{3}^0\delta_{123\omega}
\der\omega_{123},
\IEEEeqnarraynumspace
\label{eq:ppsd}
\end{IEEEeqnarray}
where $S_\omega^0\eqdef S_\omega(0)$ is the input PSD. This PSD is known as the GN model (PSD) in the literature \cite{inoue1992phase,poggiolini2012gn}. 

Note that, the signal energy is preserved in the NLS equation \eqref{eq:cnls}. However, the first-order
signal \eqref{eq:q-pert} and its consequent PSD \eqref{eq:ppsd} are
not energy preserving.

\begin{remark}
Alternatively, the GN PSD can be obtained
by simply approximating the nonlinear term
$|q|^2q$ by $|q^L|^2q^L$ in the NLS equation, 
\[
j\partial_z q=q_{tt}+2|q|^2q\approx q_{tt}+2|q^L|^2q^L, 
\]
where
$q^L_{\omega}(z)\eqdef\exp(jz\omega^2)q_{\omega}(0)$ is the solution of
the linear part of the NLS equation. 
 
\qed
\end{remark}

\begin{remark}

In the NLS equation with physical parameters \eqref{eq:nls-optics}, the GN PSD \eqref{eq:ppsd} is of order $\gamma^2$, 
where $\gamma$ is the nonlinearity coefficient. If instead of 
$q^{(1)}(\omega)$, $q^{(2)}(\omega)$ is used in \eqref{eq:s-pert}, additional terms are introduced to \eqref{eq:ppsd}. 
One of these terms is of order $\bigo{\gamma^2}$, arising from the interaction of the linear term with a nonlinear
quintic term in signal expansion. The GN PSD refers to the Fourier spectrum of the nonlinear term 
in \eqref{eq:q-pert}, ignoring
its interaction with other terms in the expansion of $q$.

\qed 
\end{remark}

\subsection{Discrete-frequency NLS Equation}
As in the continuous-frequency model, we use the solution of the
linear equation 
\begin{IEEEeqnarray}{rCl}
  q_k^{(0)}(z)=e^{j(\omega_0^2k^2-4\const{P}) z}q_k(0),
\label{eq:dnls-lin}
\end{IEEEeqnarray}
in \eqref{eq:dnls-integ} to obtain the first-order signal
\begin{IEEEeqnarray}{rCl}
  q_k^{(1)}(z)&=&e^{j(\omega_0^2k^2-4\const{P}) z}\Bigl(
q_k(0)+2jz|q_k(0)|^2q_k(0)\nn\\
&&+2 \sum\limits_{lmn\:\in\: \nr_k}
H(\Omega_{lmnk})(z)q_l(0)q_m(0)q_n^*(0)\Bigr),
\IEEEeqnarraynumspace
\label{eq:q-pert1}
\end{IEEEeqnarray}
where 
\begin{IEEEeqnarray}{rCl}
\Omega_{lmnk}\eqdef\omega_0^2(\ell^2+m^2-n^2-k^2).
\label{eq:Omega}
\end{IEEEeqnarray}
Note that $\Omega_{lmnk}\neq 0$, since singularities $l=k$ and $m=k$ have been
removed from $\mathcal{N}_k$. 

Ignoring the SPM term $2jz|q_k(0)|^2q_k(0)$ in \eqref{eq:q-pert1}, squaring and
averaging as before, the GN PSD is
\begin{IEEEeqnarray}{rCl}
  S_k^{\text{GN}}(z)=S_k^0+8\sum\limits_{n}\sum\limits_{l\neq
    k}\sum\limits_{m\neq k}|H_{lmnk}|^2S_l^0S_m^0S_n^0\delta_{lmnk}.
\IEEEeqnarraynumspace
\label{eq:d-gnpsd}
\end{IEEEeqnarray}
The cross terms between linear and nonlinear parts in
\eqref{eq:q-pert1} is zero, similar to the continuous case.

\subsection{Secular Behavior in the Signal Perturbation}
\label{sec:secular}
It can be seen that the second term in \eqref{eq:q-pert1}, corresponding to SPM, grows
unbounded with $z$.  Had the XPM not been removed, that
too would have produced a similar unbounded term. These degenerate FWM terms that tend
to infinity with $z$ are called \emph{secular terms} and make the series
divergent. As a result, \emph{regular perturbation theory} fails for the NLS equation.

The secular term of SPM can be removed using a multiple-scale
analysis. For this purpose, we introduce an additional independent slow
variable 
\[
\ell\eqdef \epsilon z,\quad q(t,z)\eqdef q(t,z,l),
\]
where now $2\epsilon \ll 1$ is the nonlinearity
coefficient. The NLS equation \eqref{eq:dnls} is transformed to
\[
 \partial_z
   q_k+\epsilon\partial_l q_k=j\omega_0^2k^2q_k-2j\epsilon\left(2\const{P}q_k-|q_k|^2q_k
+ \mathcal{N}_k\right).
\]
We expand $q_k$ in powers of $\epsilon$ and equate powers of
$\epsilon$ on both sides. We choose $\partial_lq_k=2j|q_k|^2q_k$ to remove the SPM singularity. Omitting details, 
the zero- and first-order terms \eqref{eq:dnls-lin}
and \eqref{eq:q-pert1} are, respectively, modified to 
\begin{IEEEeqnarray}{rCl}
q_k^{(0)}(z)&=&e^{j(\omega_0^2k^2-4\const{P}+2\epsilon|q_k(0)|^2)z}q_k(0),\nn\\
  q_k^{(1)}(z)&=&e^{j\left(\omega_0^2k^2-4\const{P}+2\epsilon|q_k(0)|^2\right)
    z}
\Bigl\{
q_k(0)\nn
\\
&&-\:2j\epsilon \sum\limits_{lmn\:\in\: \nr_k} H(\bar\Omega_{lmnk}) q_l(0)q_m(0)q_n^*(0) \Bigr\},
\IEEEeqnarraynumspace
\label{eq:q-pert2}
\end{IEEEeqnarray}
where
\begin{IEEEeqnarray*}{rCl}
 \bar\Omega &\eqdef&\Omega_{lmnk}+j\omega_0^2\epsilon(|q_l(0)|^2+|q_m(0)|^2-|q_n(0)|^2-|q_k(0)|^2).
\end{IEEEeqnarray*}
The PSD is given by \eqref{eq:d-gnpsd} with
$\Omega\rightarrow\bar\Omega$. 
It can be seen that the fast variable $z$ describes the rapid
evolution of $q_k$ in small distance scales. However, as $z$ is
increased, potentially important dynamics on large scales (where
$\epsilon z\approx 1$) can be
missed. In our example, the SPM term does indeed grow at scales of order $\bigo{\epsilon^{-1}}$. The role of the slow variable $l$ is to
describe dynamics at this long-haul scale. 

Secular terms seem to have been neglected in the
literature. This is because missing the sum with minus sign in
\eqref{eq:sum-or-and} ignores the SPM term in \eqref{eq:dnls}. However, typically energy is distributed over many Fourier modes
and $z|q_k(z)|^2q_k(z)$ is quite small. As a result, if $z$ is not
too large, the singular perturbation signal \eqref{eq:q-pert1} is a good
approximation and is simpler to use.

Secular terms appear in the continuous model too. Including the four terms missed in \eqref{eq:mu-123'1'2'3} gives the secular PSD 
contribution 
\begin{IEEEeqnarray}{rCl} 
S_{\omega}^{\textnormal{sec}}&=&4\int H_{123\omega}H_{1'2'3'\omega}^*\Bigl\{S_1S_{2}S_{3'}\left(\delta_{1'2}\delta_{2'3'}+\delta_{22'}\delta_{1'3'}\right)\delta_{13}\nn
\\
&&+\:
S_1S_2S_{3'}\left(\delta_{11'}\delta_{2'3'}+\delta_{12'}\delta_{1'3'}\right)\delta_{23}\Bigr\}\delta_{123\omega}\der\omega_{1231'2'3'}
\nn\\
&=&
16 z^2\const P^2 S_{\omega}.
\label{eq:S-sec}
\end{IEEEeqnarray}

Figs. \ref{fig:q-pert}(a)--(b) demonstrate the accuracy of the
first-order perturbation approximation \eqref{eq:q-pert1}. Here the strength of the 
nonlinearity is measured as the ratio $a(z)$ of the nonlinear and linear parts
of the Hamiltonian \cite{zakharov1992turb}
\begin{IEEEeqnarray*}{rCl}
 \mathcal H(z)\eqdef j\int\limits_{-\infty}^{\infty}\biggl(\underbrace{\left|\partial_t
    q(t,z)\right|^2}_{\text{linear}}-\underbrace{\left|q(t,z)\right|^4}_{\text{nonlinear}}\biggr)\der t.
\end{IEEEeqnarray*}

It can be seen in Figs. \ref{fig:q-pert}(a)--(b) that the perturbation
series rapidly diverges as $A$ is increased. Even in the pseudo-linear
regime where $a<0.1$, the error may not be small. Note that in the
focusing regime, the linear and
nonlinear parts of \eqref{eq:q-pert} add up destructively so that
$\norm{q^{(1)}_k(z)}<\norm{q_k(z)}=\norm{q_k(0)}$ and $q^{(1)}(t,1)$ is below $q(t,1)$ in Fig.~\ref{fig:q-pert}(b). However, 
in the PSD the sign is lost and the linear and nonlinear PSDs add up constructively, so that $S_k(z)$ stands above $S_k^0$ in Fig.~\ref{fig:psds}. As the
amplitude is increased, the nonlinear term grows and, regardless of
its angle, dominates the linear term. As a result, $q^{(1)}(t,z)$ goes
above $q(t,z)$ and $\norm{q^{(1)}}$ rapidly diverges to infinity. However, as we will see,
the error in the PSD is typically smaller due to the squaring
and averaging operations. 

\begin{figure}
\centering
\begin{tabular}{cc}
\includegraphics[width=4cm]{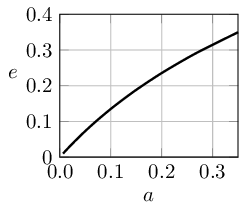}
&
\includegraphics{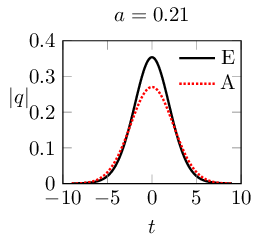}\\
~~~~~~(a) & ~~~~~~(b)
\end{tabular}
\caption{ First-order signal approximation in perturbation expansion method when
  $q(t,0)=A\exp(-t^2/2)$ and $z=1$. (a) Error $e=\norm{q-q^{(1)}}/\norm{q}$ as a function of the nonlinearity
  parameter $a$. (b) Exact (E) and approximate (A) signals when $a=0.21$ ($A=0.5$). 
}
\label{fig:q-pert}
\end{figure}


\section{Energy Transfer in the Frequency Domain}
\label{sec:energy-transfer}
In this section, we motivate the subsequent sections by explaining 
how energy is transferred among Fourier modes and why one might expect 
an asymptotically stationary PSD to the leading order in nonlinearity,
when the signal propagates 
according to the NLS equation.

We begin with a two-dimensional Fourier series restricted on the 
dispersion relation $\zeta=\omega_0^2k^2$
  \begin{IEEEeqnarray}{rCl}
    q(t,z)= \sum\limits_{k=-\infty}^{\infty}a_k(z)e^{j(k\omega_0 t+\omega_0^2k^2 z)}.
\label{eq:2d-fourier}  
\end{IEEEeqnarray}
Substituting 
\eqref{eq:2d-fourier} into the NLS equation, we get
\begin{IEEEeqnarray}{rCl}
  \partial_z a_k(z)=-2j\sum\limits_{lmn}e^{j\Omega_{lmnk}
    z}a_la_ma_n^*\delta_{lmnk},
\label{eq:da-dz}
\end{IEEEeqnarray}
where $\Omega_{lmnk}$ is defined in \eqref{eq:Omega} and the sum is
over all possible interactions $lm\rightarrow nk$.
The integrating factor $\exp(j\omega_0^2 k^2z)$ removes the additive dispersion term from the NLS equation and reveals it as an operator acting on
nonlinearity in \eqref{eq:da-dz}.  If $\Omega_{lmnk}\neq 0$ and $z$ is large, the exponential term oscillates rapidly and the
nonlinearity $a_la_ma_n^*$ is averaged out in integration over $z$, following the Riemann-Lebesgue
lemma. Therefore only modes lying on the resonant
manifold
\begin{IEEEeqnarray}{Cl}
&
\ell +m= n+k,
\IEEEyessubnumber
\label{eq:frequency}\\*
[-0.1\normalbaselineskip]
\smash{\left\{
\IEEEstrut[8\jot]
\right.} \nonumber
\\*[-0.2\normalbaselineskip]
&
\ell^2+m^2=n^2+k^2, 
\IEEEyessubnumber
\label{eq:phase}
\end{IEEEeqnarray}
contribute to the asymptotic changes in the Fourier mode $a_k$. 
This means that energy is transported in the frequency domain primary
via the resonant
interactions;  the influence of the non-resonant interactions on energy
transfer is small. The
frequency and phase matching conditions \eqref{eq:frequency} and
\eqref{eq:phase} respectively represent conservation of the energy and
momentum. 

In our example, the resonant manifold \eqref{eq:frequency}--\eqref{eq:phase} permits only trivial interactions
\begin{IEEEeqnarray}{rCl}
\{l=n, \: m=k\},\quad\text{or} \quad \{l=k, \: m=n\},  
\label{eq:triv-int}
\end{IEEEeqnarray}
describing SPM ($\ell=m$) and XPM ($\ell\neq m$). Separating out the
resonant indices from the sum in \eqref{eq:da-dz}, we get 
\begin{IEEEeqnarray}{rCl}
  \partial_z
  a_k(z)=j(-4\const{P}+2|a_k|^2)a_k-2j\mathcal{N}^{\nr}_k(a,a,a),
\label{eq:da-dz2}
\end{IEEEeqnarray}
where
\begin{IEEEeqnarray*}{rCl}
  \mathcal{N}^{\nr}_k(a,a,a)\eqdef\sum\limits_{lmn\:\in\: \nr_k}e^{j\Omega_{lmnk} z}a_la_ma_n^*\delta_{lmnk},
\end{IEEEeqnarray*}
contains only non-resonant quartets (the complement of the set \eqref{eq:triv-int}).
Non-resonant interactions constitute the majority of all
interactions, and since $\mathcal{N}^{\nr}_k\approx 0$, we
observe that, when viewed in the four dimensional space $(l,m,n,k)$,  most of the possible interactions
are nearly absent. 

Ignoring $N^{\nr}_k$ in \eqref{eq:da-dz2}, we obtain
\begin{IEEEeqnarray}{rCl}
   a_k(z)\approx j(-4\const{P}+2|a_k|^2)a_k,
\label{eq:da-dz3}
\end{IEEEeqnarray}
which does not imply any inter-modal interactions. In fact,
restoring the dispersion, we have
\begin{IEEEeqnarray*}{rCl}
q_k(z)\approx e^{j\left(\omega_0^2k^2-4\const{P}+2|q_k(0)|^2\right)z}q_k(0),
\end{IEEEeqnarray*}
which means $|q_k(z)|\approx |q_k(0)|$. This is because the resonant
quartets for the convex dispersion relation $\zeta=\omega^2$, $\omega=\omega_0k$, of the integrable NLS equation
consists of only trivial quartets \eqref{eq:triv-int}.

It follows that the signal spectrum is almost stationary. There are
small oscillations in the
spectrum due to small non-resonant effects, but because most of the possible interactions
between Fourier modes, responsible for spectral broadening, do not occur, a
localized energy stays localized and does not spread to
infinite frequencies. This also intuitively explains the lack of the equipartition,
and the periodic exchange, of the energy among Fourier modes in the 
Fermi-Pasta-Ulm (FPU) lattice \cite{yousefi2012nft1} ---  and generally in soliton systems.

Fig.~\ref{fig:two-modes} shows the evolution of modes $k=0$ and
$k=N/2$, where $N$ is the integer bandwidth, for 
input signal $q(t,0)=2\exp(-t^2/2)$  ($a(0)=5.65$) in the deterministic NLS equation. Despite local changes in distance, globally
the signal spectrum is not broadened monotonically, but rather oscillates. Here
evolution is continued for a very long distance $z=50$ (about
$10^5$ km in a standard optical system). This is not surprising given that the orbits of integrable
Hamiltonian systems in the phase space are periodic, confined to a torus. Sufficient perturbations 
to integrability break the characteristic 
oscillations in Fig.~\ref{fig:two-modes}, though for small perturbations the oscillations persist. Note that if 
the input is a stochastic process and, instead of $|q_k(z)|$, the PSD $\E |q_k(z)|^2$ is plotted, these local oscillations are further
averaged out so that the PSD is asymptotically almost stationary.  

The steady-state stationary PSD, without much transient spectral broadening, is a consequence of
integrability. Consider a non-integrable equation, \eg, by introducing
 a third-order dispersion to the NLS equation with dispersion relation
 $\zeta=\omega^3+3\omega^2$, $\omega=\omega_0k$. The resonant
manifold is
\begin{IEEEeqnarray}{Cl}
&
\ell +m= n+k,
\IEEEyessubnumber
\label{eq:resonant-example-a} 
\\*
[-0.1\normalbaselineskip]
\smash{\left\{
\IEEEstrut[8\jot]
\right.} \nonumber
\\*[-0.2\normalbaselineskip]
&
\ell^3+m^3+3(\ell^2+m^2)=n^3+k^3+3(n^2+k^2).
\IEEEyessubnumber
\label{eq:resonant-example-b} 
\end{IEEEeqnarray}
Since the dispersion relation $\zeta=\omega^3+3\omega^2$ is non-convex, 
the resonant manifold contains a larger number of quartets 
than the trivial ones in \eqref{eq:triv-int}, \eg,\ $(l, m, n, k)=(1,-3,0,-2)$.
It can be verified that non-trivial quartets are
\begin{IEEEeqnarray*}{rCl}
l+m=-2,\quad n+k=-2.
\end{IEEEeqnarray*}
 As before, ignoring $\mathcal{N}^{\nr}_k$,
equation \eqref{eq:da-dz3} now reads
\begin{IEEEeqnarray*}{rCl}
   a_k(z)\approx j(-4\const{P}+2|a_k|^2)a_k-2j\sum\limits_{\textnormal{nt}} a_la_ma_n^*,
\end{IEEEeqnarray*}
where the sum is over non-trivial quartets, \ie, the resonant quartets in \eqref{eq:resonant-example-a}--\eqref{eq:resonant-example-b}
excluding the trivial ones \eqref{eq:triv-int}. The coupling
introduced by non-trivial interactions creates a strong
energy transfer mechanism, causing substantial spectral broadening (or
narrowing, depending on the equation) and dispersing a localized energy to 
higher (lower) frequencies. Unlike the FPU lattice where energy is
exchanged periodically among a few Fourier modes, energy partitioning
continues until an equilibrium is reached. This can be a 
flat (equipartition) or non-flat stationary steady-state PSD, depending on the equation.

Note that if pulses have short duration, then $\omega_0\gg 1$ and the
dispersion operator inside the sum in $\mathcal{N}^{\nr}_k$ averages
out nonlinearity more effectively. This explains pseudo-linear
transmission in the wideband regime.

\begin{figure}
\centering
\includegraphics[scale=0.95]{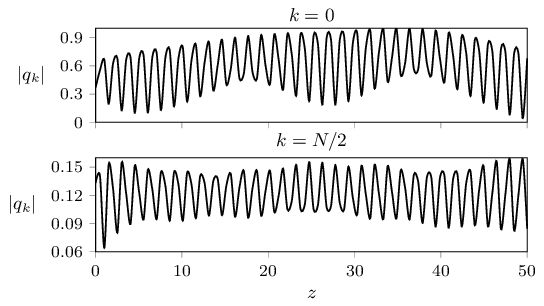}
\caption{Evolution of two Fourier modes in distance.}
\label{fig:two-modes}
\end{figure}

\begin{figure*}[t]
\centering
\includegraphics[width=\textwidth]{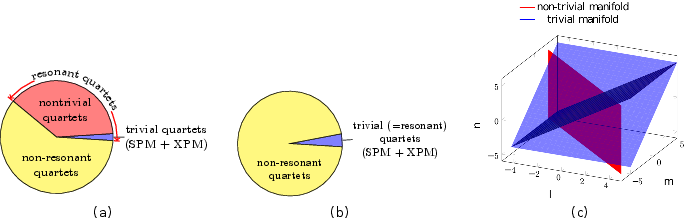}
\caption{ (a) Classification of quartets in a general nonlinear dispersive equation. When 
viewed in the four dimensional space $(l,m,n,k)$,
many quartets do not (or weakly) interact. Transfer of energy
occurs primarily among the resonant quartets. (b) In the integrable NLS
equation, the dispersion relation is convex, as a result, the resonant
quartets include only trivial quartets, which are quite sparse in the
whole space. In addition, in energy-preserving NLS equation, trivial
quartets do not interact. Only weak interactions due to non-resonant
quartets are left. (c) Resonant manifold for dispersion relation $k=\zeta^3+3\zeta^2$. Energy
flows on the red plane $l+m+2=0$. The two blue planes $n=l$ and $n=m$ form the trivial manifold. No energy flows
on the blue planes.}
\label{fig:quartets-classification}
\end{figure*}

To summarize, one can divide four-wave interactions into resonant and
non-resonant interactions. Transfer of energy takes place primarily among the
resonant modes and via the resonance mechanism. The resonant quartets
are themselves divided into trivial and
non-trivial quartets. Trivial quartets represent SPM and XPM and, in
the energy-preserving integrable NLS equation, do not cause
interaction. Non-trivial interactions, which are absent in
the integrable equation, cause coupling and transfer of energy among
all resonant modes. This
occurs when higher order dispersion or nonlinear terms are introduced
in the integrable NLS equation. The
redistribution of energy among Fourier modes continues until an equilibrium (which is generally not an equipartition) 
is reached after a transient evolution. See Fig.~\ref{fig:quartets-classification}.


\section{The KZ Model Power Spectral Density}
\label{sec:kz-model}  
In this Section, we obtain the basic KZ PSD, in a single-channel single-span
optical fiber with  no loss and higher-order dispersion terms.

\subsection{Kinetic Equation of the PSD}
\label{sec:kinetic}

We assume that the signal is strongly stationary so that
\eqref{eq:wss} holds. In particular 
\begin{IEEEeqnarray*}{rCl}
R(t_1,t_2;z)=R(\tau;z),\quad \tau \eqdef t_2-t_1.
\end{IEEEeqnarray*}
As shown in Appendix~\ref{app:cumulants-wss}, stationarity 
implies that the signal is uncorrelated in the frequency domain
\begin{IEEEeqnarray}{rCl}
\mu_{12}(z)=S_{k}(z)\delta_{12},
\label{eq:psd}
\end{IEEEeqnarray}
where $S_{k}(z)=\mathcal{F}_s(R(\tau;z))(k)$. 

Often the phase of a signal in a nonlinear
dispersive equation varies rapidly compared to the slowly-varying
amplitude. Furthermore, in some applications such as ocean waves, it is
natural to assume that the initial data is random. This suggests a
statistical approach, such as that in the turbulence theory. Here the
evolution of the $n$-point spectral cumulants is described.

The NLS equation \eqref{eq:cnls} consists of a linear term involving
$q$ and a nonlinear
term $|q|^2q$. As a result, the evolution of the $2$-point moment is
tied to the $4$-point moment, the evolution of the
$4$-point moment is tied to the $6$-point moment, and so
on. For reasons explained in Section~\ref{sec:energy-mechanisms}, we work 
with cumulants. Multivariate moments and cumulants are interchangeable via 
\eqref{eq:cumulants-to-moments} and \eqref{eq:moments-to-cumulants} in
Appendix~\ref{app:cumulants-moments}. As a result, one obtains
recursive differential equations for $2n$-point cumulants, each equation
depending on cumulants up to $(2n+2)$-point. In strongly
nonlinear systems, higher order cumulants are not negligible and the
hierarchy of cumulant equations does not truncate. This makes strong
turbulence, traditionally encountered in
solid-state physics and fluid dynamics, a difficult problem.  However, in weakly
nonlinear systems,  statistics are close to Gaussian and consequently
higher-order cumulants can be neglected. As a result, a closure of the
hierarchy of the cumulant equations is reached. This gives
rise to a kinetic equation for the PSD. In WWT, kinetic
equations can often be solved using, \eg, Zakharov conformal
transformations. The resulting solutions are known as
Kolmogorov-Zakharov spectra. 

For non-integrable equations, kinetic
equations indicate a monotonic transfer of energy to higher or lower frequencies 
(direct and reverse energy cascade) in the first order in
nonlinearity. However, for integrable equations, kinetic
equations immediately predict a stationary PSD to the first order.
Nevertheless, for the NLS equation, the kinetic equation can be solved
to the second order in the nonlinearity to account for changes in the PSD
that are observed in numerical and experimental studies of the
integrable NLS equation.

A differential equation for the 2-point moment $\mu_{kk}\eqdef S_k$ can be obtained straightforwardly: 
\begin{IEEEeqnarray}{rCl}
  \frac{\der S_k}{\der z}&=&\E\left\{ q_k^*\partial_zq_k+\cc\right\} \nn\\
&=&\E\Bigl\{ q_k^*\bigl(
j\omega_0^2 k^2 q_k-2j\sum q_lq_mq_n^*\delta_{lmnk}
\bigr)+\cc\Bigr\}\nn\\
&=&
j\omega_0^2 k^2 S_k-2j\sum  \mu_{lmnk}\delta_{lmnk}+\cc\nn\\
&=&4\sum  \Im(\mu_{lmnk})\delta_{lmnk},
\label{eq:dS-dz}
\end{IEEEeqnarray}
where $\cc$ stands for complex conjugate. To the zero order in the nonlinearity, the signal distribution is
Gaussian and $\tilde{S}_{lmnk}=0$. As a result, $\Im(\mu_{lmnk})=0$ and
$\der S_k/\der z=0$. 

In the first order in the nonlinearity, the
evolution of the 4-point moment is
\begin{IEEEeqnarray}{rCl}
\frac{\der \mu_{lmnk}}{\der z}&=&  
\E\bigl\{(\partial_z q_l)q_mq_n^*q_k^*\bigr\}+\cdots\nn\\
&=&
\E\Bigl\{\Bigl(j\omega_0^2 l^2 q_l-2j\sum\limits_{l'm'n'} q_{l'}q_{m'}q_{n'}^*\delta_{l'm'n'l}\Bigr)q_mq_n^*q_k^*
\Bigr\}\nn\\
&&+\:\cdots\nn\\
&=&j \Omega_{lmnk}\mu_{lmnk}-2j\sum\limits_{l'm'n'} \Bigl\{\mu_{l'mm'nn'k}\delta_{l'm'n'l}\nn \\
&&\:+
\mu_{ll'm'nn'k}\delta_{l'm'n'm}-\mu_{lmn'l'm'k} \delta_{l'm'n'n}
\nn\\
&&\:-
\mu_{lmn'l'm'n}\delta_{l'm'n'k}
\Bigr\}.
\label{eq:d-S1234}
\end{IEEEeqnarray}

In the discrete model, dispersion is a multiplication by a
unitary matrix. The linear and nonlinear parts of the NLS dynamics are mixing
processes in time and frequency. When the input signal is quasi-Gaussian and signal phase is uniformly
distributed in $z$, these mixing processes maintain the quasi-Gaussian distribution, in 
the view of the central limit theorem. As long as the signal phase is uniform
and nonlinear interactions are weak, this is an excellent
approximation. 

It follows that, under the assumption that there are a large number of
Fourier modes in weak interaction, and that the distribution of $q_k(0)$ is
quasi-Gaussian, we can assume that the distribution of $q_k(z)$
remains quasi-Gaussian,
as defined in \eqref{eq:quasi-gaussian}. Consequently, the four
6-point moments in \eqref{eq:d-S1234} break down 
in terms of the 2-point moments 
\begin{IEEEeqnarray*}{rCl}
 \mu_{l'mm'nn'k}&=&
 S_mS_nS_k\left(\delta_{mn'}  \delta_{l'n} \delta_{m'k}+ \delta_{mn'} \delta_{l'k}  \delta_{nm'}\right)\\
&&\:+
S_mS_nS_{n'}\left(\delta_{mk}\delta_{\ell'
    n}\delta_{m'n'}+\delta_{mk}\delta_{\ell 'n'}\delta_{m'n}\right)\\
&&\:+
S_mS_{n'}S_{k}\left(\delta_{mn}\delta_{\ell'
    n'}\delta_{m'k}+\delta_{mn}\delta_{\ell 'k}\delta_{m'n'}\right),
\end{IEEEeqnarray*}
\begin{IEEEeqnarray*}{rCl}
\mu_{l'm'lnn'k} &=&
 S_lS_nS_k\left(\delta_{ln'}\delta_{l'n}\delta_{m'k}+ \delta_{ln'} \delta_{l'k}\delta_{nm'}\right)
\\ &&\:+
 S_{\ell '}S_{n}S_k\left(\delta_{l'n'}\delta_{m'n}\delta_{\ell
     k}+\delta_{l'n'}\delta_{m'k}\delta_{\ell n}\right)
\\ && \:+
S_{m'}S_{n}S_{k}\left(\delta_{m'n'}\delta_{\ell n}\delta_{\ell' k}+\delta_{m'n'}\delta_{\ell k}\delta_{\ell' n}\right),
\end{IEEEeqnarray*}
\begin{IEEEeqnarray*}{rCl}
\mu_{lmn'l'm'k} &=&
 S_lS_mS_k\left( \delta_{n'k} \delta_{ll'}\delta_{mm'}+ \delta_{n'k} \delta_{lm'}\delta_{ml'}\right)
\\
&&+\: S_lS_mS_{n'}\left(\delta_{n'\ell '}\delta_{\ell m'}\delta_{mk}+\delta_{n'\ell'}\delta_{lk}\delta_{mm'}\right)
\\
 &&+\:S_kS_{\ell'}S_{n'}\left(\delta_{n'm '}\delta_{\ell \ell'}\delta_{mk}+\delta_{n'm'}\delta_{lk}\delta_{m\ell'}\right),
\end{IEEEeqnarray*}
\begin{IEEEeqnarray*}{rCl}
\mu_{lmn'l'm'n} &=&
 S_lS_mS_n\left(\delta_{n'n} \delta_{ll'}\delta_{mm'}+\delta_{n'n} \delta_{lm'}\delta_{ml'}\right)
\\
&&\:+
 S_lS_mS_{n'}\left(\delta_{n'\ell'}\delta_{\ell
     m'}\delta_{mn}+\delta_{n'\ell'}\delta_{\ell n}\delta_{mm'}\right)
\\
&&
\:+ S_lS_mS_{n'}
\left(\delta_{m'n'}\delta_{\ell\ell'}\delta_{mn}+\delta_{m'n'}\delta_{\ell n}\delta_{mn'}\right).
\end{IEEEeqnarray*}

Summing over primed variables in \eqref{eq:d-S1234}, the first two terms in
the four expressions above
add up to $4jT_{lmnk}\delta_{lmnk}$, where 
\begin{IEEEeqnarray*}{rCl}
  T_{lmnk}(S,S,S)(z)&\eqdef& S_lS_mS_n+S_lS_mS_k \\
&& -\: S_lS_nS_k-S_mS_nS_k
\end{IEEEeqnarray*}
is the \emph{collision term}. The last four terms in the four expressions simplify to zero in \eqref{eq:d-S1234}.
Canceling
$\delta_{lmnk}$ in the resulting equation, it follows that
\begin{IEEEeqnarray}{rCl}
  \frac{\der S_{lmnk}}{\der z}&=&j\Omega_{lmnk}S_{lmnk}+4jT_{lmnk}(S,S,S)(z),
\label{eq:dSlmnk-dz}
\end{IEEEeqnarray}
where, recall that $\mu_{lmnk}=S_{lmnk}\delta_{lmnk}$.

In the standard WTT approach, it is assumed that $S_{lmnk}$ varies
slowly. As a result, $\der S_{lmnk}/\der z\approx 0$ in
\eqref{eq:dSlmnk-dz}, thus
$S_{lmnk}=-4T_{lmnk}/\Omega_{lmnk}$. If $\Omega_{lmnk}=0$, $S_{lmnk}$ can not be determined from 
\eqref{eq:dSlmnk-dz}. Replacing $\Omega_{lmnk}$ with $\Omega_{lmnk}-j\epsilon$ and 
using the Kramers-Kronig relations \cite[Lemma 1]{yousefi2015nfdm}, we 
get $\Im(\frac{1}{x-j\epsilon})=-\pi\delta(x)$ in the sense of distributions. This gives
$\Im(S_{lmnk})=4\pi T_{lmnk}/\Omega_{lmnk}\delta(\Omega_{lmnk})$ and, subsequently, the
standard kinetic equation for the NLS equation
\begin{IEEEeqnarray*}{rCl}
  \frac{\der S_k}{\der z}=16\pi\sum \frac{T_{lmnk}}{\Omega_{lmnk}}\delta(\Omega_{lmnk})\delta_{lmnk}.
\end{IEEEeqnarray*}
The product of the two delta functions dictates resonant (trivial) interactions
\eqref{eq:triv-int}. This means energy transfer occurs primarily among
resonant modes. However, for resonant interactions $T_{lmnk}=0$, and a
stationary spectrum is obtained. The stationarity of
the turbulence spectrum of integrable systems is discussed in \cite{zakharov2009tis}. 

However, it can be seen in \eqref{eq:dSlmnk-dz} that even if
$T_{lmnk}$ is slowly varying, \eg, $T_{lmnk}=0$, $S_{lmnk}$ oscillates
with spatial frequency $\Omega_{lmnk}$ for non-resonant quartets, for
which $\Omega_{lmnk}\neq 0$. 
This linear dynamics modulates the collision term in \eqref{eq:dSlmnk-dz}. Since we are interested in non-stationary spectrum,
we cannot assume $\der S_{lmnk}/\der z\approx 0$, and evolution of 
$S_{lmnk}$, due to non-resonant interactions, has to be accounted for in the next
order. This is very easy to perform and has been pointed out in
\cite{suret2011wave} as well. 

The integral form of \eqref{eq:dSlmnk-dz} is
\begin{IEEEeqnarray}{rCl}
 S_{lmnk}(z)&=&e^{j\Omega_{lmnk}z}S_{lmnk}(0)\nn\\
&&+\:4j\int\limits_0^z
  e^{j\Omega_{lmnk}(z-z')} T_{lmnk}(S,S,S)(z')\der z'.
\IEEEeqnarraynumspace
\label{eq:dSlmnk-dz-integ}
\end{IEEEeqnarray}

Since resonant interactions \eqref{eq:triv-int} do
not contribute to $\der S_k/\der z$, below we include only non-resonant
interactions for which $\Omega_{lmnk}\neq 0$ and $S_{lmnk}(0)=\tilde{S}_{lmnk}(0)$.
Substituting \eqref{eq:dSlmnk-dz-integ} into \eqref{eq:dS-dz}, we
obtain the kinetic equation for $S_k$
\begin{IEEEeqnarray}{rCl}
  \frac{\der S_{k}}{\der
    z}&=&4\sum\limits_{lmn\:\in\: \nr_k}\Im\left(e^{j\Omega_{lmnk}z}\tilde S_{lmnk}(0)\right)\nn\\
&&\hspace{-1cm}+\:16\epsilon^2
\sum\limits_{lmn\:\in\: \nr_k} \int_0^z
\cos(\Omega_{lmnk}(z-z'))T_{lmnk}(z')\der z',
\IEEEeqnarraynumspace
\label{eq:kinetic}
\end{IEEEeqnarray}
where parameter $\epsilon$ is introduced to use it below.
 
The kinetic equation \eqref{eq:kinetic} is a nonlinear cubic equation
similar to the NLS equations. However, now the 
rapidly-varying variables are averaged out and the PSD evolves 
very slowly so that the perturbation theory is better applicable.
We thus solve \eqref{eq:kinetic} perturbatively, writing
\begin{IEEEeqnarray*}{rCl}
  S_k(z)&=&S_k^{(0)}(z)+\epsilon S_k^{(1)}(z)+\cdots,
\\
  S_{lmnk}(z)&=&S_{lmnk}^{(0)}(z)+\epsilon S_{lmnk}^{(1)}(z)+\cdots.
\end{IEEEeqnarray*}
For the zero-order term we obtain
\begin{IEEEeqnarray*}{rCl}
S_k^{(0)}(z)=S_k^0+4\Re\Bigl(\sum\limits_{lmn\:\in\: \nr_k} H_{lmnk}\tilde
S_{lmnk}(0)\Bigr).
\end{IEEEeqnarray*}
If the input signal is quasi-Gaussian, $\tilde S_{lmnk}(0)\approx 0$
and the contribution of the 
second term to the PSD can be typically
ignored. Consequently, we can substitute $S_k^{(0)}(z)=S_k^0$ in the
equation of the next order. Omitting details, we obtain 
\begin{IEEEeqnarray}{rCl}
  S_k^{\text{KZ}}(z)=S_k^0&+& 8 \epsilon^2\sum\limits_{\substack{l\neq k\\m\neq
    k}}|H_{lmnk}(z)|^2 T_{lmnk}^0\delta_{lmnk}.
\label{eq:wtpsd}
\end{IEEEeqnarray}

Note that if $z\rightarrow\infty$, $S_k^{\text{KZ}}$ is \emph{approximately} stationary. Equation \eqref{eq:wtpsd} is the
KZ PSD.

\subsection{KZ Model Assumptions}
In this subsection,  we summarize the assumptions of the KZ model and comment on
their validity in the context of fiber-optic data communications.

\paragraph{Fourier transforms $q_k$ and $S_k$ exist}
 Particularly,
  $R(\tau;0)$ should 
vanish as $|\tau|\rightarrow \infty$. 

This assumption is valid in data
communications because signals have finite energy and time duration.

\paragraph{The input signal is strongly stationary}
This ensures that the $2n$-point moments are concentrated on stationary
manifolds. In particular, $q_k$ are uncorrelated, as stated in \eqref{eq:psd}. 
The delta functions that follow from this assumption simplify the
collision term in \eqref{eq:d-S1234}.

This assumption is valid in uncoded OFDM systems,
  where sub-carrier symbols are independent and the transmitted signal
  is cyclostationary. However, in coarse WDM systems
  the time-domain pulse shape can make the transmitted signal non-stationary and
  cause correlations in the frequency domain.

\paragraph{Signal has quasi-Gaussian distribution for all $z$ in the sense of \eqref{eq:quasi-gaussian}} 

In particular the input signal must be quasi-Gaussian. Under random
phase approximation \cite{zakharov1992turb}, the flow of the NLS
equation would then ensure that
the signal remains quasi-Gaussian in the weak nonlinearity
framework. This assumption is needed in
\eqref{eq:4-point}--\eqref{eq:6-point} to close the cumulant equations.
 
The integrable NLS equation
in the focusing regime has stable soliton solutions. As pointed out in
\cite{zakharov2004owt}, the solitonic regime, in which the
nonlinearity is strong, can act against 
the dispersive mixing of the weak nonlinearity regime. We assume that for
random input the coherence is not developed. This means that the
interference spectrum in the focusing and defocusing regimes are the same.

To summarize, Assumptions b) and c) may fail in data
communications. However, the WWT approach can be
re-worked out without using these assumptions. The price to pay is that the closure
is achieved at orders above six (see
\eqref{eq:4-point}--\eqref{eq:6-point}) and the expressions are not as simple. In Section \ref{sec:wdm-application},
we obtain the KZ spectrum for a WDM input signal with and without
Assumptions b) and c).


\section{Comparing the KZ and GN Models}
\label{sec:kz-vs-gn}

In this section we explain how the KZ model differs from the GN model.

\subsection{Differences in Assumptions}

The GN model assumes a perfectly Gaussian distribution compared with
the less stringent quasi-Gaussian assumption of the KZ model. Note that in the
presence of the four-wave interactions $lm\rightarrow nk$, higher
order moments are encountered. If a closure is to be reached, any
perturbative method requires reducing high-order moments to low-order ones, \ie, the quasi-Gaussian
assumption at some order. For example, the breakdown of the 6-point moments is also required in
the GN model, in closing \eqref{eq:s-pert} for the 2-point moment.

The GN PSD in some scenarios has been modified to account for a fourth-order non-Gaussian
noise \cite{dar2013pnp} (see Remark~\ref{rem:correction-terms}). Its perturbation expansion can also be carried
out to higher orders to improve the accuracy and account for deviations
from the Gaussian distribution. However, given the same assumptions,
the GN and KZ PSDs are still different. Furthermore, to calculate moments methodically,
one ends up using WWT framework anyways.

\subsection{Differences in PSD}
To connect the KZ and GN models, we wrote the modified kinetic
equation and the KZ spectrum \eqref{eq:wtpsd} in terms of the
same kernel $H_{lmnk}$ that appears in
the GN model. As a result, from \eqref{eq:wtpsd} it can be readily seen that
\begin{IEEEeqnarray*}{rCl}
  S^{\text{KZ}}_k(z)=S^{\text{GN}}_k(z)-\Delta S_k(z),
\end{IEEEeqnarray*}
where
\begin{IEEEeqnarray*}{rCl}
\Delta S_k&\eqdef&8S_k^0\sum\limits_{lmn\:\in\: \nr_k}|H_{lmnk}(z)|^2\left(S_l^0S_n^0+S_m^0S_n^0-S_l^0S_m^0\right).  
\end{IEEEeqnarray*}
That is to say,  the KZ PSD modifies the GN PSD by subtracting $\Delta S_k$
from it. That makes the KZ PSD at any order $n$ in perturbation expansion as accurate as GN
PSD at order $n+1$.  The improvement might be small in current systems
operating near the pseudo-linear regime, however, as the signal amplitude is
increased the GN PSD rapidly diverges from the true PSD.

The KZ PSD is energy-preserving unlike the GN PSD. Perturbation
expansion in signal breaks the structure of the NLS equation, so that
some important features of the exact equation can be
lost. For example, the average signal power according to the GN PSD is
\begin{IEEEeqnarray*}{rCl}
  \const{P}(z)=\const{P}(0)+\int\limits_{-\infty}^\infty
  |H_{123\omega}|^2S_1^0S_2^0S_{3}^0\delta_{123\omega}\der\omega_{123\omega}.
\end{IEEEeqnarray*}
It is seen that the signal
power is not preserved (see also Fig. \ref{fig:q-pert}(b)). This is
because, at any
order in perturbation, ignoring the energy of the higher-order terms
breaks energy conservation.
 
In contrast, in the KZ model, noting the symmetries 
\begin{IEEEeqnarray}{rCl}
H_{lmnk}=H_{mlnk}=H_{lmkn},\quad |H_{lmnk}|=|H_{nklm}|,  
\label{eq:symmetries}
\end{IEEEeqnarray}
and the similar ones for $\delta_{lmnk}$, we have
\begin{IEEEeqnarray*}{rCl}
\sum\limits_{\substack{lmnk}}|H_{lmnk}|^2
 S_l^0S_m^0S_n^0\delta_{lmnk}
=\sum\limits_{\substack{lmnk}}|H_{lmnk}|^2 S_l^0S_n^0S_k^0\delta_{lmnk},
\IEEEeqnarraynumspace
\end{IEEEeqnarray*}
where we substituted $lmnk\leftrightarrow nklm$. 
It follows that $\sum\limits_k S_k^{\text{KZ}}(z)=\sum\limits_k S_k(0)$,
\ie, the KZ model is energy-preserving. Other conservation laws exist
for kinetic equations \cite{zakharov1992turb}.

Fig.~\ref{fig:psds} compares the power spectral density of the GN
and KZ models. Here the input is a zero-mean Gaussian process with
$S^0_k=A^2\exp(-\omega_0^2 k^2)$ with $A=3\sqrt{2\pi/N}$,
$\omega_0=2\pi/N$, $N=2048$. The simulated output PSD is 
measured at $z=1$ over 10000 input instances in a single-channel NLS equation. Despite being in the nonlinear regime ($a(1)=4.29$),
the KZ PSD still approximates the simulated PSD remarkably well. Note
that $S_k^{\text{KZ}}$ crosses the $S^0_k$ curve so that it has the same
area, while the GN model PSD is well
above both the $S^0_k$ and the simulated PSD. Therefore the GN model is pessimistic, predicting a
higher interference than the actual one.

For the GN PSD to converge, a small power ($\sim$ 0.5 \rm{mW})
has to be distributed over a large bandwidth so that $\norm{q_k}\ll
1$ and the cubic term in $S_k^{\text{GN}}$ does not grow.

\begin{figure}
\centering
\includegraphics[width=8cm]{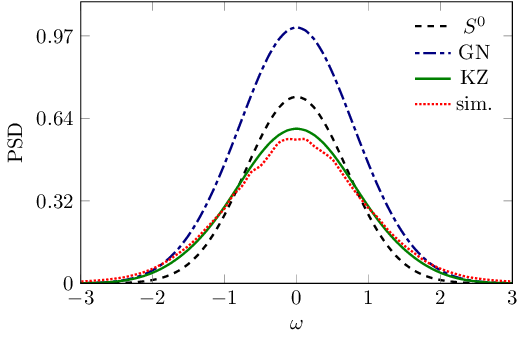}
\caption{Comparing the (normalized) PSD of the GN and KZ models with
  the simulated (sim.) PSD.}
\label{fig:psds}
\end{figure}

\subsection{Differences in Energy Transfer Mechanisms}
\label{sec:energy-mechanisms}
Since the GN model assumes a Gaussian distribution, only $\Re(\mu_{lmnk})$ and  $\Re(\mu_{lmnrpq})$ are responsible for spectrum
evolution. In contrast, changes in the KZ PSD stem merely from
$\Im(\mu_{lmnk})$. 
Because of the important factor $j$ in the NLS equation,
$\Re(\mu_{lmnk})$ does not contribute to changes in PSD. Arbitrary 
non-Gaussian statistics can occur along $\Re(\mu_{lmnk})$
without impacting the PSD.  Under the assumptions of the GN model (that the probability 
distribution is Gaussian in evolution),
$\Im(\mu_{lmnk})=0$ and the KZ model (correctly) predicts a
stationary spectrum. Consequently, deviations from Gaussianity are necessary for
any spectral change.

A problem with the signal perturbation, and consequently with
the GN model, is that here one works with moments, not
the cumulants as in the KZ model. Moments of a scalar Gaussian random variable $X$,
are $\E|X|^{2n}=(\E|X|^2)^{n}(2n-1)!!$, where $n!!=n(n-2)(n-4)\cdots$, which grow with $n$.
Higher-order moments cannot be ignored in the analysis. 
In contrast, higher-order cumulants are zero for Gaussian distribution and as the amplitude is
increased, they are gradually generated sequentially in increasing
order. The fact that cumulants are centered around a Gaussian distribution
makes them suitable for use in a perturbation theory around the
linear solution.

\subsection{Differences in Probability Distributions}
From the previous discussion, it follows that signal distribution in
the KZ model is a zero-mean non-Gaussian distribution with the following moments:
Asymmetric moments are zero; 2-point moment is given by the KZ PSD
\eqref{eq:wtpsd}; 4-point moment is given by
\eqref{eq:dSlmnk-dz-integ}; other higher order moments are given in terms of the $2$- and $4$-point
moments according   to \eqref{eq:cumulants-to-moments}, with zero $2n$-point cumulants, $n\geq 3$.


\section{Application to WDM}
\label{sec:wdm-application}

One application of the PSD is to estimate the interference power in
WDM systems. In models where the XPM amounts to a constant phase shift, the
interference at frequency $k$ is non-degenerate FWM, as well as part of the
degenerate FWM; see
Fig.~\ref{fig:quartets} and Example~\ref{ex:interference}. However, in
the WDM
literature often the whole FWM is treated as interference. That is to
say, all of the nonlinearity $\mathcal{N}_k$ in the NLS equation
\eqref{eq:dnls} is treated as noise. The corresponding spectra $S_k^{\text{KZ}}$ and $S_k^{\text{GN}}$
include self- and cross-channel interference.

Consider a WDM system with $2N+1$ users, each
having bandwidth $\Omega_s$.  In WDM the following 
(baseband) signal is sent over the channel
\begin{IEEEeqnarray}{rCl}
  q(t,0)=\sum\limits_{m=-N}^{N}\left(\sum\limits_{l=1}^Ma_m^l\phi^l(t)\right)e^{jm\Omega_s
  t},
\label{eq:wdm-signal}
\end{IEEEeqnarray}
where $l$ and $m$ are time and user indices, $\Omega_s$ is the user bandwidth, and $\phi^l(t)$ is an
orthonormal basis for the space of finite-energy $T$-periodic signals with Fourier
transform in $[-\Omega_s/2 , \Omega_s/2]$. Finally, $a_m^l$ is a sequence of complex-valued random
variables, independent between users, but potentially correlated within each user, \ie, 
\begin{IEEEeqnarray*}{rCl}
  \E a_m^la_{m'}^{l'*}= \mu^{ll'}_{mm}(a)\delta_{mm'},
\end{IEEEeqnarray*}
where $\mu^{ll'}_{mm}(a)\eqdef\E a_m^la_m^{l'*}$ is the symbols correlation
function of the
user $m$, due to, \eg, channel coding. The set of frequencies of the user
$m$, $-N\leq m\leq N$, is 
\begin{IEEEeqnarray*}{rCl}
  A_{m}=\Bigl\{ m\Omega_s+k\Omega_0 \:|\: -N_0/2\leq k< N_0/2\Bigr\},
\end{IEEEeqnarray*}
where $N_0\eqdef\lfloor\Omega_s/\Omega_0\rfloor$ and $\Omega_0\eqdef 2\pi/T$.

\subsection{Stationary Gaussian WDM Signals}
In this case, the assumptions of the GN and KZ models are satisfied.
 The interference ``spectrum'' is
\begin{IEEEeqnarray*}{rCl}
  S_k^{\text{NL}}=8\sum\limits_{lmn\:\in\: \nr_k} |H_{lmnk}|^2T_{lmnk}, 
\end{IEEEeqnarray*}
where for the KZ model $T_{lmnk}$ is the collision term, and for the
GN model $T_{lmnk}=S_lS_mS_n$. The intra (self)-channel interference
for the central user $m=0$ is the part of the sum in
$S_k^{\text{NL}}$ where $l,m,n\in A_0$. This is somewhat
similar to SPM. The rest
of terms, where at least one index is in the complement set $\bar A_0$, is the inter-channel
interference.  This is divided into three parts: 1) exactly two
indices are in $A_0$ (1-wave interference) 2) exactly one index is in
$A_0$  (2-wave interference) 3) no index is in $A_0$ (3-wave
interference). The 1-wave interference has fewer terms than the
others and can be ignored. The 2-wave interference is akin to XPM but
is not similarly averaged out and should be accounted for. 

Note that the net interference is zero in the KZ
model, \ie, $S_{k}^{\text{NL}}$ is negative for some $k$. 
 
\subsection{Non-stationary non-Gaussian WDM Signals}

The correlation function of the WDM signal \eqref{eq:wdm-signal} is
\begin{IEEEeqnarray*}{rCl}
  R(t_1,t_2)&=&\sum\limits_{mll'}\mu_{mm}^{ll'}(a)\phi^l(t_1)\phi^{l'*}(t_2)\exp(-jm\Omega_s(t_2-t_1))\\
&\overset{(a)}{=}&\const{P}_0\sum\limits_{l=1}^M\phi^l(t_1)\phi^{l*}(t_2)E(t_2-t_1),
\end{IEEEeqnarray*}
where $E(x)\eqdef\sum_m\exp(-jm\Omega_s x)$ and step $(a)$ follows under the additional assumption that $a_m^l$
is \iid, so that
$\mu_{mm}^{ll'}=\const{P}_0\delta_{ll'}$, $\const{P}_0\eqdef\E|a_m^l|^2$. Unless  in special cases, \eg,
$\phi^l(t)=\exp(jl\Omega_0 t)$, the input  signal is not a 
stationary process. This can be seen in the frequency domain too. The Fourier series
coefficients are
\begin{IEEEeqnarray}{rCl}
  q_k=\sum\limits_{lm}a_m^l\phi_{mN_0+k}^l,
\label{eq:qk}
\end{IEEEeqnarray}
where 
\begin{IEEEeqnarray*}{rCl}
\phi^l_k= 
 \begin{cases}
\mathcal{F}_s(\phi^l(t))(k), & -N_0/2\leq k<N_0/2, \\
0, & \text{otherwise}.
\end{cases}
\end{IEEEeqnarray*}
The orthogonality of $\phi^l(t)e^{jn\Omega_s t}$ and $\phi^{l'}(t)e^{jn'\Omega_s t}$ in the frequency domain reads
\begin{IEEEeqnarray*}{rCl}
  \sum\limits_{k=-\frac{N_0}{2}}^{\frac{N_0}{2}-1}\phi_{k+nN_0}^{l}\phi_{k+n'N_0}^{*l'}=\delta_{ll'}\delta_{nn'}.
\end{IEEEeqnarray*}
The 2-point spectral moment at $z=0$ is
\begin{IEEEeqnarray}{rCl}
  \mu_{12}&=&\const{P}_0\sum\limits_{lm}\phi_{mN_0+k_1}^l\phi_{mN_0+k_2}^{*l}.
\label{eq:correlations}
\end{IEEEeqnarray}
In general $\mu_{12}\neq \mu_{11}\delta_{12}$, unless in special cases, \eg, if $k_1$ and $k_2$ belong to
two different users, or
$\phi^l_k=\phi_0\delta_{k}$, or $\phi^l_k=\phi_0\exp(j\frac{2\pi}{M}kl)$.
 
In addition to the stationarity Assumption b), Gaussianity Assumption c) may also not hold in WDM. In 
particular the input distribution is arbitrary.  For non-Gaussian inputs, the cumulant
$\kappa_{123456}(0)$ should be included. 

The correlations and non-Gaussian input statistics can be introduced into 
the GN and KZ models using $\mu_{ij}$ and $\kappa_{123456}$. 
Repeating the analysis in the paper, the GN and KZ PSDs in WDM are
 \begin{IEEEeqnarray*}{rCl}
  S_k^{\text{GN}}&=&S^0_k+4\sum\limits_{\nr_{k}, \nr_{k'}}H_{lmnk}H_{l'm'n'k}\mu_{lmn'l'm'n},\\
S_k^{\text{KZ}}&=&S^0_k+8\sum\limits_{\nr_k}|H_{lmnk}|^2T_{lmnk},
\end{IEEEeqnarray*} 
where
\begin{IEEEeqnarray}{rCl}
  T_{lmnk}&\eqdef&\frac{1}{2}
\sum\limits_{l'm'n'}\Bigl(
\mu_{l'mm'nn'k}\delta_{l'm'n'k}+
\mu_{ll'm'nn'k}\delta_{l'm'n'm}\nn\\
&&-\:
\mu_{lmn'l'm'k}\delta_{l'm'n'n}-
\mu_{lmn'l'm'n}\delta_{l'm'n'k}
\Bigr).
\label{eq:Tlmnk}
\end{IEEEeqnarray}

The 6-point moment is 
\begin{IEEEeqnarray}{rCl}
\mu_{123456}
&=&
\sum\limits_{l_{1-6}m_{1-6}}\mu_{123456}^{123456}(a)\phi_{m_1N_0+k_1}^{l_1}\phi_{m_2N_0+k_2}^{l_2}\phi_{m_3N_0+k_3}^{l_3}
\nn\\
&&\qquad\qquad\times\:
\phi_{m_4N_0+k_4}^{*l_4}\phi_{m_5N_0+k_5}^{*l_5}\phi_{m_6N_0+k_6}^{*l_6}\nn\\
&=&
\mu_{14}\mu_{25}\mu_{36}+
\mu_{14}\mu_{26}\mu_{35}+
\mu_{15}\mu_{24}\mu_{36}\nn\\
&&+\:
\mu_{15}\mu_{26}\mu_{34}+
\mu_{16}\mu_{24}\mu_{35}+
\mu_{16}\mu_{25}\mu_{34}\nn\\
&&+\: \kappa_{123456},
\label{eq:mu123456}
\end{IEEEeqnarray}
where  $\mu_{ij}$ is given in \eqref{eq:correlations}, and we used 
\begin{IEEEeqnarray*}{rCl}
  \mu_{123456}^{123456}(a)&\eqdef&\E a_{m_1}^{l_1}a_{m_2}^{l_2}a_{m_3}^{l_3}
a_{m_4}^{*l_4}a_{m_5}^{*l_5}a_{m_6}^{*l_6}\\
&=&
\const{P}^3\Bigl(\delta_{14}^{14}\delta_{25}^{25}\delta_{36}^{36} +
\delta_{14}^{14}\delta_{26}^{26}\delta_{35}^{35}+
\delta_{15}^{15}\delta_{24}^{24}\delta_{36}^{36}\\
&&+\:
\delta_{15}^{15}\delta_{26}^{26}\delta_{34}^{34}+
\delta_{16}^{16}\delta_{24}^{24}\delta_{35}^{35}+
\delta_{16}^{16}\delta_{25}^{25}\delta_{34}^{34}
\Bigr)\\
&&+\:\kappa_{123456}^{123456}(a),
\end{IEEEeqnarray*}
where $\delta_{m_1m_2}^{l_1l_2}\eqdef\delta_{m_1m_2}\delta_{l_1l_2}$. 

For \iid\ symbols, from
\eqref{eq:moment-to-cumulant-4} and \eqref{eq:moment-to-cumulant-6}, we have
$\kappa_{12}^{12}(a)=\tilde{S}_{12}(a) \delta_{12}^{12}$,
$\kappa_{1234}^{1234}(a)=\tilde{S}_{1234}(a)\delta_{12}^{12}\delta_{23}^{23}\delta_{34}^{34}$,
$\kappa_{123456}^{123456}(a)=\tilde{S}_{123456}(a)
\delta_{12}^{12}\delta_{23}^{23}\delta_{34}^{34}\delta_{45}^{45}\delta_{56}^{56}$,
and so on, with cumulant densities
\begin{IEEEeqnarray}{rCl}
  \tilde{S}_{12}(a) &=&\E|a|^2,\nn\\
 \tilde{S}_{1234}(a)&=&\E|a|^4-2\E^2|a|^2 ,\label{eq:4-cumulant}\\
 \tilde{S}_{123456}(a)&=&\E|a|^6-9\E|a|^2\E|a|^4+12\E^3|a|^2\label{eq:6-cumulant}.
 \end{IEEEeqnarray}

The simplifications of Section~\ref{sec:kinetic} in the case of
uncorrelated Gaussian signals,
due to integration
over delta functions, do not occur anymore. If $\phi^l(t)=p(t-lT/M)$,
where $p(t)$ is a pulse shape in time interval $[0,T/M]$, then $\phi_k^l=p_k\exp(j2\pi kl/M)$ and
\begin{IEEEeqnarray*}{rCl}
  \mu_{12}&=&\const{P}(2N+1)p_1p_2^*\sum\limits_{l=1}^M\exp(j2\pi
  l(k_1-k_2)/M)\\
&=&\const{P}(2N+1)Mp_1p_2^*\delta_{12}.
\end{IEEEeqnarray*}
In this case there is no correlation and PSDs are modified only via $\kappa_{123456}$.

\begin{remark}
The accuracy of the GN model has been improved in the enhanced GN model (EGN) 
\cite{dar2013pnp, serena2015,poggiolini2015simple,carena2014egn}. In the EGN model, correction terms are 
introduced to the GN model to account for non-Gaussianity. The forth-order 
correction term in \cite{dar2013pnp} is identified with the cumulant \eqref{eq:4-cumulant} in the KZ model. Likewise, 
the correction
terms $\Phi_a$ and $\Psi_a$ in \cite[Eq. 6]{carena2014egn} are, respectively, identified with 
cumulants \eqref{eq:4-cumulant} and \eqref{eq:6-cumulant}.
Furthermore, KZ 
model illustrates how infinitely many such terms can be added methodically. 

\label{rem:correction-terms}
\end{remark}

\subsection{Phase Interference}
The power spectral density of the nonlinear term in the NLS equation does not suggest that
one should consider nonlinearity as additive noise. In fact, the PSD obviously does
not capture cross-phase interference. In the energy-preserving NLS equation, XPM is a
constant phase shift, as shown, \eg, in \eqref{eq:dnls}.  However, in WDM, from
\eqref{eq:qk}, the signal energy is
\begin{IEEEeqnarray}{rCl}
  \sum\limits_{k=1}^{(2N+1)N_0}|q_k|^2=\sum\limits_{lm}|a_m^l|^2\phi_{mN_0+k}^l\phi_{mN_0+k}^{*l}.
\label{eq:energy-sum}
\end{IEEEeqnarray}
Typically, the per-user power is a known constant,
however, in an optical mesh network, the power of interfering users
may not be known. Energy is also not preserved in the presence of
loss. In such cases where XPM is no longer a constant phase shift,
part of the sum \eqref{eq:energy-sum} where $m\neq 0$ acts as cross-phase interference
for the center user.
 
The interference resulting from XPM is discussed in \cite{mecozzi2012nsl}. This is done by
substituting the WDM  input signal \eqref{eq:wdm-signal} into the
approximate solution \eqref{eq:q-pert}, sorting out interference
terms, and naming XPM and FWM. 


\section{KZ and GN PSDs in Multi-Span Systems}
\label{sec:multi-span}
The PSDs \eqref{eq:ppsd} and \eqref{eq:wtpsd} hold for one span of
lossless fiber with second-order dispersion. In this section, we
include loss and higher order dispersion, and generalize
\eqref{eq:ppsd} and \eqref{eq:wtpsd}
to multi-span links with amplification.  

We consider a multi-span optical system with $N$ spans, each of length
$\epsilon$, in a fiber of total 
length $z$, $z=N\epsilon$. 
Pulse
propagation in the overall link is governed
by 
\begin{IEEEeqnarray}{rCl}
  \partial_z
  q_\omega(z)&=&j\left(\frac{j\alpha(z)}{2}-\beta(\omega)\right)q_\omega(z)-j\gamma\mathcal{N}_\omega(q,q,q)(z)\nn\\
&&+\:\left(\sum\limits_{n=1}^N \frac{G_n (z)}{2}\delta(z-n\epsilon)\right)q_\omega(z)
,
\label{eq:nls-optics}
\IEEEeqnarraynumspace
\end{IEEEeqnarray}
where $\alpha(z)$ is (power) loss exponent, $G_n(z)\eqdef\int_{(n-1)\epsilon}^{n\epsilon} \alpha (l)\der l$
is the lumped gain exponent at the end of span $n$, $\gamma$ is the nonlinearity
coefficient and
\begin{IEEEeqnarray*}{rCl}
  \beta(\omega)=\beta_0+\beta_1(\omega-\omega_0)+\frac{\beta_2}{2}(\omega-\omega_0)^2+\cdots,
\end{IEEEeqnarray*}
is the dispersion function (also known as the wavenumber or propagation
constant). 

Lumped power amplification at the end of each span restores the linear part
of PSD, however, since loss is distributed, it does not normalize the
nonlinear part. As a result, signal amplification leads to a growth of FWM
interference, which we calculate in this section.

\begin{remark}
Loss and periodic amplification have been discussed in \cite{churkin2015} in the context of fiber
lasers. Here a modified kinetic 
equation approach is taken to describe 
laser spectrum.  Changes in spectrum (kinetics) in \cite{churkin2015} occur due to loss and 
periodic amplification, \ie, perturbations to integrability. In the transmission problem, on the other hand, there is kinetics 
even with no loss and amplification in integrable model; see \eqref{eq:wtpsd}, as well as numerical simulations 
of the actual PSD in the literature of the GN model, and \cite{suret2011wave}. In our problem the hypothesis of delta concentration $\delta(\Omega_{123\omega})$
of  the standard turbulence \cite{zakharov1992turb} does not hold with desired 
accuracy. For parameters where 
the models of \cite{churkin2015} and this paper 
coincide, the observations are in agreement. In this section, we generalize \eqref{eq:wtpsd}.
\end{remark}

\subsubsection{GN Model}
Consider the NLS equation \eqref{eq:nls-optics} with loss, dispersion
$\beta(\omega)$, and amplification. Let
\begin{IEEEeqnarray*}{rCl}
 F(z)\eqdef\int_0^z\left(\alpha(l)- \sum\limits_{n=1}^N G_n
   (l)\delta(l-n\epsilon)\right)\der l.
\end{IEEEeqnarray*}
Comparing \eqref{eq:nls-optics} with the 
dimensionless NLS equation, we identify $\omega^2z\rightarrow jF(z)/2
-\beta(\omega)z$. Therefore
\begin{IEEEeqnarray*}{rCl}
  \Omega_{123\omega}z=\left(\omega_1^2+\omega_2^2-\omega_3^{*2}-\omega^2\right)z\\
\end{IEEEeqnarray*} 
in signal \eqref{eq:q-pert} and PSD \eqref{eq:wtpsd} is replaced with 
\[
\bar\Omega_{123\omega}z=jF(z)-\Omega_{123\omega}z,
\]
where now
\begin{IEEEeqnarray*}{rCl}
\Omega_{123\omega}\eqdef\beta(\omega_1)+\beta(\omega_2)-\beta(\omega_3)-\beta(\omega).
\end{IEEEeqnarray*} 
The GN PSD \eqref{eq:ppsd} is modified to
\begin{IEEEeqnarray*}{rCl}
  S_\omega(z)=e^{-F(z)}\left(S_\omega(0)+2\gamma^2\int 
|\tilde H_{123\omega}|^2 S_1S_2S_3\delta_{123\omega}\der\omega_{123}\right),
\end{IEEEeqnarray*}
where 
\begin{IEEEeqnarray}{rCl}
  \tilde{H}_{123\omega}&\eqdef&-j\int\limits_0^z
  e^{j\bar\Omega_{123\omega}l}\der l\nn\\
&=&-j\int\limits_0^z
  e^{-(F(l)+j\Omega_{123\omega}l)}\der l.
\IEEEeqnarraynumspace
\label{eq:gn-htilde}
\end{IEEEeqnarray}

Several cases can be derived from \eqref{eq:gn-htilde}. 

\paragraph{Single-span lossy fiber}
In a single-span fiber with constant loss $\alpha$ and no amplification, $G=0$ and $F(z)=\alpha z$. Thus $\tilde
H_{123\omega}=H_{123\omega}(j\alpha
-\Omega_{123\omega})(z)$. This shows the effect of loss and higher order dispersion.

\paragraph{Multi-span links}
In a multi-span link with constant loss $\alpha$, 
\begin{IEEEeqnarray}{rCl}
 F(l)=\alpha
l-\alpha\epsilon\sum_{i=1}^NU(l-i\epsilon)=\alpha(l-n\epsilon),\:
n=\lfloor l/\epsilon\rfloor, 
\IEEEeqnarraynumspace
\label{eq:Fl}
\end{IEEEeqnarray}
where $U(x)$ is the
Heaviside step function. Thus
\begin{IEEEeqnarray}{rCl}
  \tilde{H}_{123\omega}&=&-j\int\limits_0^z
  e^{\left(-\alpha l+\alpha\epsilon\sum\limits_{n=1}^N
      U(l-n\epsilon)-j\Omega_{123\omega}l\right)}\der l\nn\\
&=&-j\sum\limits_{n=0}^{N-1}\int\limits_{n\epsilon^-}^{(n+1)\epsilon^-}
e^{\left(-\alpha (l-n\epsilon)-j\Omega_{123\omega}l\right)}\der l\nn\\
&=&-j \sum\limits_{n=0}^{N-1}e^{-jn\epsilon\Omega_{123\omega}}\int\limits_{0}^{\epsilon}
e^{j\left(j\alpha -\Omega_{123\omega}\right)z'}\der z'\nn\\
&=&H(j\alpha-\Omega_{123\omega})(\epsilon)G_{123\omega}^{\text{GN}},
\label{eq:ms-gn-psd}
\end{IEEEeqnarray}
 where
\begin{IEEEeqnarray}{rCl}
  G_{123\omega}&\eqdef&\sum\limits_{n=0}^{N-1}e^{-jn\epsilon\Omega_{123\omega}}=\frac{1-e^{-jz
      \Omega_{123\omega}}}{1-e^{-j\epsilon\Omega_{123\omega}}}\nn\\
&=&e^{-j\frac{\Omega}{2}(z-\epsilon)}
\frac{\sin(z\Omega_{123\omega}/2)}{\sin(\epsilon\Omega_{123\omega}/2)}.
\label{eq:G123omega}
\end{IEEEeqnarray}

Therefore,  the GN PSD of multi-span link is given by the same
equation \eqref{eq:ppsd}, with $H(\Omega_{123\omega})(z)$ replaced with
$(\gamma/2)H(j\alpha-  \Omega_{123\omega})(\epsilon)G_{123\omega}$. Note that
$\tilde H$ describes the FWM growth in both the signal and PSD. For
further clarification, see Appendix~\ref{app:gn-multi-span}.

\subsubsection{KZ Model}
Considering the analysis of Section~\ref{sec:kinetic}, factors $-(\der
F(z)/\der z) S_k$ and 
$-2(\der F(z)/\der z) S_{lmnk}$ appear, respectively, in
the right hand sides of moment equations \eqref{eq:dS-dz}  and \eqref{eq:dSlmnk-dz}. The KZ PSD is
\begin{IEEEeqnarray}{rCl}
  S_{k}(z)&=&e^{-F(z)}\Bigl(S_k^0
+\:2\gamma\sum\limits_{\nr_k}\int\limits_{0}^ze^{F(z')}\Im(S_{lmnk}(z'))\der z'\Bigr),
\IEEEeqnarraynumspace
\label{eq:dSk-dz-2}
\end{IEEEeqnarray} 
where
\begin{IEEEeqnarray}{rCl}
  \Im(S_{lmnk}(z))&=&2\gamma\int\limits_0^z
e^{-2 (F(z)-F(z'))}\cos\left(\Omega_{lmnk}(z-z')\right)
\nn\\
&&\qquad\times\:T_{lmnk}(z')\der z'.
\IEEEeqnarraynumspace
\label{eq:dSlmnk-dz-2}
\end{IEEEeqnarray}
Here, as in \eqref{eq:dSlmnk-dz-integ} and \eqref{eq:kinetic}, 
we assumed that $S_{lmnk}(0)$ is
real-valued for quasi-Gaussian input. 

We substitute \eqref{eq:dSlmnk-dz-2} into \eqref{eq:dSk-dz-2} and
solve the resulting fixed-point equation iteratively starting from $S_k(z)=0$. The first
iterate gives $S_k^{(0)}(z)=\exp(- F(z))S_k^0$. In the next iterate, the collision term is found
to be $T_{lmnk}(z)=\exp(-3 F(z))T_{lmnk}(0)$, which is no longer
constant. Using this collision term in  \eqref{eq:dSlmnk-dz-2}, and
subsequently in \eqref{eq:dSk-dz-2}, we obtain 
\begin{IEEEeqnarray}{rCl}
  S_k^{\text{KZ}}
=e^{-F(z)}\left(S_k^0+2\gamma^2\sum\limits_{lmn\:\in\: \nr_k}|\tilde H_{lmnk}|^2 T_{lmnk}\right),
\label{eq:kzpsd-multi-span}
\end{IEEEeqnarray}
where
\begin{IEEEeqnarray}{rCl}
  |\tilde H_{lmnk}|^2&=&2\int\limits_0^z\int\limits_{0}^{z'} e^{-(F(z')+F(l))}
  \cos\left(\Omega(z'-l)\right)\der l\der
  z'.
\IEEEeqnarraynumspace
\label{eq:Htilde-kz}
\end{IEEEeqnarray}
The integration in \eqref{eq:Htilde-kz} is over a triangle. However
the function under integration is symmetric in $l$ and $z'$, \ie,
around the line $l=z'$. Thus integration can be extended to the rectangle:
 \begin{IEEEeqnarray}{rCl}
  |\tilde H_{lmnk}|^2&=&\int\limits_0^z\int\limits_{0}^{z} e^{-(F(z')+F(l))}
  \cos\left(\Omega(z'-l)\right)\der l\der
  z'\nn\\
&=&
\left|
\int\limits_0^z e^{-(F(l)+j\Omega_{lmnk} l)}\der l\right|^2.
\end{IEEEeqnarray}

This is the same as $|\tilde H_{123\omega}|^2$ in \eqref{eq:gn-htilde}
for the GN model, with $123\omega\rightarrow lmnk$. 

It follows that in all cases, the
PSD kernels $|H_{123\omega}|^2$ and $|\tilde{H}_{123\omega}|^2$ in the GN
and KZ models are the same.

\paragraph{Single-span lossy fiber}
\noindent
For constant $\alpha$, the PSD is given by \eqref{eq:kzpsd-multi-span}
with $F(z)=\alpha z$ and
\begin{IEEEeqnarray*}{rCl}
|\tilde H_{lmnk}|^2&=&|H(j\alpha-\Omega_{lmnk})(z)|^2\\
&=&2e^{-\alpha z}\Bigl(
\frac{\cosh(\alpha z)-\cos(\Omega z)}{\alpha^2+\Omega^2}
\Bigr).
\end{IEEEeqnarray*}

This shows that in the presence of loss and physical parameters, just
as in the GN model, the
KZ PSD, after amplification $\exp(\alpha z)$ at the end of the link, is the same as \eqref{eq:wtpsd}, with $H(\Omega_{lmnk})(z)$ replaced
with $(\gamma/2)H(j\alpha-\Omega_{lmnk})(z)$.

\paragraph{Multi-span links}

In the multi-span link, at the end of the link $F(z)=0$. 
The PSD is given by \eqref{eq:kzpsd-multi-span} with $F(z)=0$ and
\begin{IEEEeqnarray*}{rCl}
\tilde{H}_{123\omega}=H(j\alpha-\omega_{123\omega})(\epsilon) G_{123\omega}.
\end{IEEEeqnarray*}
with the same $G_{123\omega}$ given by \eqref{eq:G123omega}.


\section{Conclusions}

A mathematical framework based on the WWT theory is presented to study
the evolution of multi-point cumulants in nonlinear dispersive
partial differential equations with random input data. This framework
is used to explain how energy is distributed among Fourier
modes in the nonlinear Schr\"odinger equation, by
considering interactions among four Fourier modes and studying
the role of the resonant, non-resonant, and trivial quartets
in the dynamics. As an application, a PSD, termed KZ model, is proposed for 
calculating the interference power in WDM systems.

The GN model, often used in optical communication, suggests a spectrum evolution, in agreement
with numerical and experimental fiber-optic
transmissions. That seemingly conflicts with the WWT which predicts a stationary
spectrum for integrable models. It is shown that if the kinetic equation of WWT is solved to the next order in
nonlinearity, a PSD is obtained which is similar to the GN PSD.
The two models are explained mathematically and connected with each
other. The analysis shows that the kinetic
equation of the NLS equation better describes the PSD. This is not
surprising, because the GN model applies perturbation theory to the signal
equation, while the WWT applies it directly to the
PSD differential equation. The assumptions of the WWT are verified in
data communications and the basic KZ PSD is extended to various cases
encountered in communications. The GN model is also simplified for
clarity and comparison.


\section*{Acknowledgments}
The author thanks Frank Kschischang and Gerhard Kramer 
for many comments and in-depth discussions. The ideas were gradually crystallized in discussion 
with them.  Their contributions substantially 
improved this research.


\appendices


\section{Moments and Cumulants}
\label{app:cumulants}
\subsection{A Property of the Stationary Processes}
\label{app:cumulants-wss}
We make use of the following simple lemma throughout the
paper, which says that the spectral moments of a stationary process 
are supported on the stationary manifold \eqref{eq:stationary-manifold}.

\begin{lemma}
If  $q(t)$ is a strongly stationary stochastic process with finite power and existing Fourier transform, then 
$\mu_{1\cdots 2n}=S_{1\cdots 2n}\delta_{1\cdots 2n}$.
\label{lemma:correlation}
\end{lemma}

\begin{proof}
Define
\begin{IEEEeqnarray*}{rCl}
E(t_1,\cdots, t_{2n})\eqdef\exp\bigl(j(s_1\omega_1t_1+\cdots+s_{2n}\omega_{2n}t_{2n})\bigr).
\end{IEEEeqnarray*}
Shifting $t_i$  by $t_1$, it can be verified that
\begin{IEEEeqnarray}{rCl}
E(t_1,\cdots, t_{2n})=\exp(j\Delta\omega t_1) E(0, t_2-t_1,\cdots, t_{2n}-t_1),
\nn\\
\label{eq:shift-property}
\IEEEeqnarraynumspace
\end{IEEEeqnarray}
where $\Delta\omega\eqdef\sum_{i=1}^{2n}s_i\omega_i$. Expressing $q_i$ in \eqref{eq:n-point-psd} using the inverse Fourier transform, the
(strong) stationarity property implies that 
\begin{IEEEeqnarray*}{rCl}
 \mu_{1\cdots  2n} &=&\E \Bigl(q_1\cdots q^*_{2n}\Bigr)\nn\\ 
&=&\int R(t_1,\cdots, t_{2n})E(t_1,\cdots, t_{2n})\der t_{1-2n}
\nn\\
&=&\int R(0,t_2-t_1,\cdots, t_{2n}-t_{1})\\
&&\times E(0,t_2-t_1,\cdots, t_{2n}-t_1)\exp(j\Delta\omega t_1)\der t_{1-2n}\\
&=&\delta(\Delta\omega)
\int R(0,\tau_2,\cdots, \tau_{2n})E(0,\tau_2,\cdots, \tau_{2n})\der \tau_{2-2n}\\
&=&S_{1\cdots 2n}\delta(\Delta\omega),
\end{IEEEeqnarray*}
where $\tau_i\eqdef t_i-t_1$ and
\begin{IEEEeqnarray*}{rCl}
S_{1\cdots 2n}\eqdef\mathcal{F}(R(0,\tau_2,\cdots, \tau_{2n}))(0, s_2\omega_2,\cdots, s_{2n}\omega_{2n}).
\end{IEEEeqnarray*}

\end{proof}

The lemma essentially follows from the shift property of the exponential function 
\eqref{eq:shift-property}.  A more general statement is the Wiener-Khinchin theorem.

As a corollary, if $q(t)$ is stationary, $\mu_{\omega\omega}$ is infinity.

\subsection{Cumulants}

\label{app:cumulants-moments}
Let $q=(q_1,\cdots,q_{n})$ be a complex-valued random vector. We
define the joint moment 
generating function of $q$, $\Phi:\Complex^n\mapsto\Reals$, as
\begin{IEEEeqnarray}{rCl}
  \Phi(\zeta)&\eqdef&
\E\exp\left(\Re(\zeta^H q)\right)
\label{eq:Phi}
\\
&=&
\sum\limits_{r,s=0}^\infty \frac{1}{r!s!}\mu^r_{s}\zeta^r\zeta^{*s},
\label{eq:Phi2}
\end{IEEEeqnarray}
where $r,s$ are $n$-dimensional multi-indices, $r!=\prod_{k=1}^{n}r_k!$, $r_k=0,1,\cdots$, $\zeta^r=\prod_{k=1}^n\zeta_k^{r_k}$, and
\begin{IEEEeqnarray*}{rCl}
\mu^r_s=\E q^{*r}q^{s},
\end{IEEEeqnarray*}
is the $|r+s|$-point moment. It can be verified that
\begin{IEEEeqnarray}{rCl}
  \mu^r_s=\frac{\partial^{r+s}}{\partial
    \zeta^{r}\zeta^{*s}}\Phi(\zeta)\Bigr|_{\zeta=0}.
\label{eq:mu-rs}
\end{IEEEeqnarray}

The joint cumulant generating function is
$\Psi(\tau)=\log\Phi(\tau)$. The cumulants $\kappa_s^r$ are defined
from multi-variate Taylor expansion similar to \eqref{eq:Phi2} and
\eqref{eq:mu-rs}.

Note that with the notation of Section \ref{sec:notation}, $\mu_{1\cdots
  2n}=\mu_{r}^{s}$ with $|r|=|s|=n$. Cumulants $\kappa$ relate to cumulant densities $\tilde S$ via
\begin{IEEEeqnarray*}{rCl}
  \kappa_{1\cdots 2n}=\tilde S_{1\cdots 2n}\delta_{1\cdots 2n}.
\end{IEEEeqnarray*}

Joint moments can be obtained from joint cumulants by applying chain
rule of differentiation to $\Phi=\exp\Psi$, obtaining 
\begin{IEEEeqnarray}{rCl}
  \mu_{1\cdots 2n}=\sum\limits_{p\in P}\prod\limits_{a\in p}\kappa_{a},
\label{eq:cumulants-to-moments}
\end{IEEEeqnarray}
where $P$ is the set of all partitions of $(1,2,\cdots, 2n)$.  The expression is
considerably simplified for a zero-mean process, which is assumed
throughout this paper. Further simplifications occur by noting that
the asymmetric moments and cumulants  are zero. With these simplifications,
the first four relations are: 
\begin{IEEEeqnarray}{rCl}
  \mu_1 &=&\kappa_1=0,\nn\\
\mu_{12} &=& \kappa_{12},\nn\\
 \mu_{1234} &=& \kappa_{13}\kappa_{24}+ \kappa_{14}\kappa_{23}+\kappa_{1234},\nn\\
\mu_{123456} &=&
\Bigl(
\kappa_{14}\kappa_{25}\kappa_{36}+\kappa_{14}\kappa_{26}\kappa_{35}+
\kappa_{15}\kappa_{24}\kappa_{36}\nn\\
&&+\:\kappa_{15}\kappa_{26}\kappa_{34}+
\kappa_{16}\kappa_{24}\kappa_{35}+\kappa_{16}\kappa_{24}\kappa_{35}\Bigr)\nn\\
&&+\:\Bigl(\kappa_{14}\kappa_{2356}+\kappa_{15}\kappa_{2346}+\kappa_{16}\kappa_{2345}\nn\\
&&+\:\kappa_{24}\kappa_{1234}+\kappa_{25}\kappa_{1356}+\kappa_{26}\kappa_{1345}\nn\\
&&+\:\kappa_{34}\kappa_{1256}+\kappa_{35}\kappa_{1246}+\kappa_{36}\kappa_{1245}
\Bigr)\nn\\
&&+\:\kappa_{123456}.
 \label{eq:cumulant-to-moment-6}
\end{IEEEeqnarray}
For a stationary process $\kappa_{ij}=\mu_{ij}=S_{ii}\delta_{ij}$ and we obtain
\eqref{eq:4-point}. Since higher-order cumulants are
smaller than the second-order cumulant for quasi-Gaussian distributions, we can assume 
$\kappa_{lmnk}=\kappa_{lmnl'm'n'}=0$, thereby obtaining \eqref{eq:6-point}.

Cumulants can be obtained from moments by applying M\"obius inversion formula
to \eqref{eq:cumulants-to-moments}, obtaining
\begin{IEEEeqnarray}{rCl}
  \kappa_{1\cdots 2n}=\sum\limits_{p\in P}\prod_{a\in p} (-1)^{|p|-1} (|p|-1)!\mu_{a},
\label{eq:moments-to-cumulants}
\end{IEEEeqnarray}
where $|p|$ is the number of the sets in the partition $p$, \ie, the number of
products in $\mu_a$. Setting asymmetric
moments to zero, we have
\begin{IEEEeqnarray}{rCl}
\kappa_{12} &=& \mu_{12}\nn\\
 \kappa_{1234} &=& -\mu_{13}\mu_{24}- \mu_{14}\mu_{23}+\mu_{1234}\label{eq:moment-to-cumulant-4}\\
\kappa_{123456} &=&
2\Bigl(
\mu_{14}\mu_{25}\mu_{36}+\mu_{14}\mu_{26}\mu_{35}+
\mu_{15}\mu_{24}\mu_{36}\nn\\
&&+\:\mu_{15}\mu_{26}\mu_{34}+
\mu_{16}\mu_{24}\mu_{35}+\mu_{16}\mu_{25}\mu_{34}\Bigr)\nn\\
&&-\:\bigl(\mu_{14}\mu_{2356}+\mu_{15}\mu_{2346}+\mu_{16}\mu_{2345}\nn\\
&&+\:\mu_{24}\mu_{1234}+\mu_{25}\mu_{1356}+\mu_{26}\mu_{1345}\nn\\
&&+\:\mu_{34}\mu_{1256}+\mu_{35}\mu_{1246}+\mu_{36}\mu_{1245}
\Bigr)\nn\\
&&+\:\mu_{123456}.
\label{eq:moment-to-cumulant-6}
 \end{IEEEeqnarray}
For \iid\ zero-mean random variables, any variables matching in
$\mu_{1\cdots 2n}$ is canceled by terms prior to $\mu_{1\cdots 2n}$ in
\eqref{eq:moment-to-cumulant-6} and \eqref{eq:moment-to-cumulant-4},
except when all variables are equal. Thus
$\kappa_{1234}=\left(\E|q_k|^4-2\E^2|q_k|^2\right)\delta_{12}\delta_{23}\delta_{34}$ and so on.


\section{GN PSD in Multi-span Links}
\label{app:gn-multi-span}

In Section \ref{sec:multi-span}, the multi-span PSDs were obtained in a
unified manner by introducing function $F(z)$ and modifying kernels
$H_{123\omega}$. One consequence is that multi-span PSDs can (expectedly) be
obtained from single-span PSDs, regardless of whether the interference is added
coherently or not in the signal picture.  The multi-span GN PSD \eqref{eq:ms-gn-psd} is
known in the literature \cite{poggiolini2012gn}. It is presumably obtained in the manner
described below; however, it is often intuitively explained rather than 
fully derived. 

 At the end of the first span, after amplification, we have
\begin{IEEEeqnarray*}{rCl}
  q_{\omega}(\epsilon)&=&
e^{-j\epsilon\beta(\omega) }\Bigl\{
q_{\omega}(0)
-j\gamma N_{\omega}(q,q,q|H)(0,\epsilon)
\Bigr\},
\end{IEEEeqnarray*}
where 
\begin{IEEEeqnarray*}{rCl}
  N_{\omega}(q,q,q|H)(z,z')&\eqdef&\int
  H_{123\omega}(j\alpha-\Omega_{123\omega})(z'-z)\\
&&\times\:q_1(z)q_2(z)q_3^*(z)\der\omega_{123}.
\end{IEEEeqnarray*}
At the end of the second span
\begin{IEEEeqnarray*}{rCl}
  q_{\omega}(2\epsilon)&=&e^{-j\epsilon\beta(\omega)}\Bigl\{
q_{\omega}(\epsilon)-j\gamma N_{\omega}(q,q,q|H)(\epsilon,2\epsilon)
\Bigr\}\nn\\
&=&e^{-j2\epsilon\beta(\omega)} q_{\omega}(0)-j\gamma
e^{-j2\epsilon\beta(\omega)}N_{\omega}(q,q,q|H)(0,\epsilon)\\
&&-j\gamma e^{-j\epsilon\beta(\omega)} N_{\omega}(q,q,q|H)(\epsilon,2\epsilon).
\end{IEEEeqnarray*}
The last term contains $q_{\omega}(\epsilon)$, which itself is the sum
of a linear and a nonlinear term. In agreement with the first-order
approach of the GN model, all FWM terms are evaluated at the linear solution;
the contribution of the nonlinear term to $q_{\omega}(2\epsilon)$ is
of second order $\gamma^2$. We thus evaluate the last term at $\exp(-j\epsilon\beta(\omega))q_{\omega}(0)$:
\begin{IEEEeqnarray*}{rCl}
 && \mathcal{N}_{\omega}(q,q,q|H)(\epsilon,2\epsilon)=\int
  H_{123\omega}(\epsilon)
q_1(\epsilon)q_2(\epsilon)q_3^*(\epsilon)
\delta_{123\omega}\der\omega_{123}\\
&=&e^{-j\epsilon\beta(\omega)}
\int
e^{-j\Omega \epsilon}
  H_{123\omega}(\epsilon)
q_1(0)q_2(0)q_3^*(0)
\delta_{123\omega}\der\omega_{123}\\
&=&e^{-j\epsilon\beta(\omega)} N_{\omega}(q,q,q|e^{-j\Omega\epsilon}H)(0,\epsilon).
\end{IEEEeqnarray*}
Thus
\begin{IEEEeqnarray*}{rCl}
  q_{\omega}(2\epsilon)&=&e^{-2j\epsilon\beta(\omega)}\Bigl\{
q_{\omega}(\epsilon)-j\gamma  N_{\omega}\left(q,q,q|G_1H\right)(0,\epsilon)\Bigr\},
\end{IEEEeqnarray*}
where $G_1\eqdef 1+e^{-j\Omega\epsilon}$. By induction, we have
\begin{IEEEeqnarray}{rCl}
   q_{\omega}(z)&=&e^{-jz\beta(\omega)}\Bigl\{
q_{\omega}(0)-j\gamma
N_{\omega}\left(q,q,q|GH\right)(0,\epsilon)\Bigr\},
\IEEEeqnarraynumspace
\label{eq:q-pert-multispan}
\end{IEEEeqnarray}
where $G$ is given by \eqref{eq:G123omega}. Squaring and averaging 
\eqref{eq:q-pert-multispan}, we obtain the
multi-span GN PSD.


\bibliographystyle{IEEEtran}


\end{document}